\documentclass[11pt]{article}
\usepackage{etoolbox}
\newtoggle{full}
\toggletrue{full}

\iftoggle{full}{%
\usepackage[margin=1in]{geometry}
}{%
\usepackage[margin=0.96in]{geometry}
}

\usepackage{amsmath}
\usepackage{amsthm}
\usepackage{amssymb}
\usepackage{enumerate}
\usepackage{hyperref}
\newtheorem{theorem}{Theorem}[section]

\newtheorem{lemma}[theorem]{Lemma}
\newtheorem{claim}[theorem]{Claim}

\newcommand{\R}{\mathbb{R}}
\newcommand{\trivi}{\hfill$\square$}
\newcommand{\st}{fatness{~}}
\newcommand{\konstlin}{2n+3m}
\newcommand{\konstexp}{m}
\newcommand{\konstterm}{2n+6m}
\newcommand{\ketenn}{2n+3m}

\title{Concave Generalized Flows with Applications to Market Equilibria\\
\nottoggle{full}{%
\vspace{8mm}
\normalsize{Extended abstract}
}{}
}

\author{L\'aszl\'o A. V\'egh\thanks{Supported by NSF Grant CCF-0914732.}\\
{College of Computing, Georgia Institute of Technology}\\
E-mail: {\tt veghal\char'100 cs.elte.hu } 
}

\begin{document}
\nottoggle{full}{%
\begin{titlepage}
\thispagestyle{empty}
}

\maketitle

\begin{abstract}
We consider a nonlinear extension of the generalized network flow model, with the flow leaving an arc being an increasing concave function of the flow entering it, as proposed by Truemper \cite{Truemper78} and Shigeno \cite{Shigeno06}. We give a polynomial time combinatorial algorithm for solving corresponding flow maximization problems, finding an $\varepsilon$-approximate solution in $O(m(m+\log n)\log(MUm/\varepsilon))$ arithmetic operations and value oracle queries, where $M$ and $U$ are upper bounds on simple parameters. 
This also gives a new algorithm for linear generalized flows, an efficient, purely scaling variant
of the Fat-Path algorithm by Goldberg, Plotkin and Tardos \cite{Goldberg91}, not using any cycle cancellations.

We show that this general convex programming model serves as a common framework for several market equilibrium problems, including the linear Fisher market model and its various extensions. Our result immediately extends these market models to more general settings. We also obtain a combinatorial algorithm for nonsymmetric Arrow-Debreu Nash bargaining, settling an open question by Vazirani \cite{Vazirani11}. 
\nottoggle{full}{%

\medskip

The full version  of this paper is available on arXiv ({\tt http://arXiv.org/abs/1109.3893}).
}{}
\end{abstract}

\nottoggle{full}{%
\end{titlepage}
}

\section{Introduction}
A classical extension of network flows is the 
 {\sl generalized network flow model}, with a {gain factor} 
$\gamma_e>0$ associated with each arc $e$ so that if $\alpha$ units of flow enter arc $e$, then $\gamma_e\alpha$ units leave it. Since first studied in the sixties by Dantzig \cite{Dantzig63} and Jewell \cite{Jewell62}, the problem has found many applications including financial analysis, transportation, management sciences, see \cite[Chapter 15]{amo}.

In this paper, we consider a nonlinear extension, {\sl concave generalized flows},  studied by Truemper \cite{Truemper78} in 1978, and by Shigeno \cite{Shigeno06} in 2006. For each arc $e$ we are given a concave, monotone increasing  function $\Gamma_e$ such that if $\alpha$ units enter $e$ then $\Gamma_e(\alpha)$ units leave it. 
We give a combinatorial algorithm for corresponding flow maximization problems, with running time polynomial in the network data and some simple parameters.
We also exhibit new applications, showing that it is a general framework containing multiple convex programs for market equilibrium settings, for which combinatorial algorithms have been developed in the last decade. 
As an application, we also get a combinatorial algorithm for nonsymmetric Arrow-Debreu Nash bargaining, resolving an open question by Vazirani \cite{Vazirani11}. We can also extend existing results to more general settings.

Generalized flows are linear programs and thus can be solved efficiently by general linear programming techniques, the currently most efficient such algorithm being the interior-point method by Kapoor and Vaidya \cite{Kapoor96}. Combinatorial approaches have been used since the sixties (e.g. \cite{Jewell62, Onaga66,Truemper77}), yet the first polynomial-time combinatorial algorithms were given only in 1991
by Goldberg, Plotkin and Tardos \cite{Goldberg91}. This inspired a line of research to develop further
polynomial-time combinatorial algorithms, e.g.\ \cite{Cohen94,Goldfarb96,Goldfarb97,Tardos98,Fleischer02,Goldfarb02a,Goldfarb02,Wayne02,Radzik04,Restrepo09};  for a survey on combinatorial generalized flow algorithms,   see \cite{Shigeno07}. Despite the vast literature, no strongly polynomial algorithm is known so far. %
Our algorithm for this special case derives from the \textsc{Fat-Path} algorithm in \cite{Goldberg91}, with the remarkable difference that no cycle
cancellations are needed.

Nonlinear extensions of generalized flows have also been studied, e.g. in \cite{Ahlfeld87,Bertsekas97},  minimizing a separable convex cost function for generalized flows. However, these frameworks do not contain our problem, which involves nonlinear convex constraints. 


Concave generalized flows being nonlinear convex programs, they can also be solved  by the ellipsoid method, yet no practically efficient methods are known for this problem. Hence finding a combinatorial algorithm is also a matter of running time efficiency. 
Shigeno \cite{Shigeno06} gave the first combinatorial algorithm that runs in polynomial time for some restricted classes of functions $\Gamma_e$, including piecewise linear. It is also an extension of the \textsc{Fat-Path} algorithm in \cite{Goldberg91}. In spite of this development, it has remained an open problem to find a combinatorial polynomial-time algorithm for arbitrary  concave increasing gain functions.

Our result settles this question by allowing arbitrary increasing concave gain functions provided via value oracle access.
The running time bounds for this general problem are reasonably close to the most efficient linear generalized flow algorithms. 
Concave gain functions extend the applicability range of the classical generalized flow model, as they can  describe e.g. deminishing marginal utilities. From the application  point of view, another contribution of the paper is extending generalized flow techniques to the domain of market equilibrium computations, where this model turns out to be a concise unifying framework.

The concave optimization problem might have irrational optimal solutions: in general, we give a fully polynomial-time approximation scheme, with running time dependent on $\log (\frac 1\varepsilon)$ for finding an
$\varepsilon$-approximate solution. In the market equilibrium applications we have rational convex programs (as in \cite{Vazirani11}): the existence of a rational optimal solution is guaranteed. We show a general technique to transform a sufficiently good approximation delivered by our algorithm to an exact optimal solution under certain circumstances. \iftoggle{full}{We demonstrate how this technique can be applied on the example of nonsymmetric Arrow-Debreu Nash bargaining, where the existence of a combinatorial algorithm was open \cite{Vazirani11}. 

In Section~\ref{sec:preliminaries}, we give the precise definition of the problems considered. Thereby we introduce a new, equivalent variant of the problem, called the {\sl symmetric formulation}, providing a more flexible algorithmic framework. 
Section~\ref{sec:market-app} shows the applications for market equilibrium problems.
Section~\ref{sec:background} explores the background of minimum-cost circulation and generalized flow algorithms, and exhibits the main algorithmic ideas. 
We first present our symmetric generalized flow algorithm in 
Section~\ref{sec:gen-alg} for the special case of linear gains. Based on this, Section~\ref{sec:concave} gives the algorithm for arbitrary concave gain functions. Section~\ref{sec:sink} adapts these algorithms for the more standard sink formulation of the problems. Section~\ref{sec:market} considers the case when  the existence of a rational optimal solution is guaranteed, and shows how the approximate solution provided by our algorithm can be turned to an optimal solution. The final Section~\ref{sec:disc} discusses possible further directions.}
{%

The rest of this extended abstract is organized as follows. 
In Section~\ref{sec:preliminaries}, we give the precise definition of the problems considered. 
Thereby we introduce a new, equivalent variant of the problem, called the {\sl symmetric formulation}, providing a more flexible algorithmic framework. 
Section~\ref{sec:market-app} shows the applications for market equilibrium problems.
Section~\ref{sec:background} explores the background of minimum-cost circulation and generalized flow algorithms.
Section~\ref{sec:concave} gives the algorithm for symmetric concave generalized flows. Section~\ref{sec:sink} shows how the algorithm can be applied for the more standard sink formulation. 

The full version of the paper first presents the algorithm for the special case of linear generalized flows. It also describes a general method for finding the optimal solutions for rational convex programs, and applies it to the  nonsymmetric Arrow-Debreu Nash bargaining problem,  where the existence of a combinatorial algorithm was open \cite{Vazirani11}. Furthermore, it gives a detailed overview on related linear and convex flow problems.
}{}

\section{Problem definitions}\label{sec:preliminaries}
\iftoggle{full}{We define two closely related variants of the linear and the concave generalized flow problem.
Let $G=(V,E)$ be a directed graph. 
Let $n=|V|$, $m=|E|$, and for each node $i\in V$, let $d_i$ be the total number of incoming and outgoing arcs incident to $i$.

We are given lower and upper arc capacities $\ell,u:E\rightarrow \R$  and gain factors $\gamma:E\rightarrow \R^+$ on the arcs, and node demands
$b:V\rightarrow \R$.
By a {\sl pseudoflow} we mean a function  $f:E\rightarrow \R$ with $\ell\le f\le u$.
Given the pseudoflow $f$, let 
\begin{equation}\label{eq:excess}
e_i:=\sum_{j:ji\in E}\gamma_{ji}f_{ji}-\sum_{j:ij\in E}f_{ij}-b_i.
\end{equation}

In the first variant of the problem, called the \textsl{sink formulation},
there is a distinguished {\em sink} node $t\in V$. 
The objective is to maximize  $e_t$ for pseudoflows satisfying $e_i\ge 0$ for all $i\in V-t$.

This differs from the  way the problem is usually defined in the literature  with the more restrictive $e_i=0$ for $i\in V-t$, and assuming $\ell\equiv 0$, $b\equiv 0$. However, this problem can easily be reduced to solving the sink formulation, see e.g. \cite{Shigeno07}.

The following extension has been proposed by Truemper \cite{Truemper78} and Shigeno \cite{Shigeno06}.
On each arc $ij\in E$, we are given lower and upper arc capacities $\ell,u:E\rightarrow \R$  and a monotone increasing concave function $\Gamma_{ij}:[\ell_{ij},u_{ij}]\rightarrow \R\cup\{-\infty\}$; we are also given node demands $b:V\rightarrow \R$. 
As for generalized flows, a pseudoflow is a function $f:E\rightarrow \R$ with $\ell\le f\le u$. For a pseudoflow $f$, let 
\[
e_i:=\sum_{j:ji\in E}\Gamma_{ji}(f_{ji})-\sum_{j:ij\in E}f_{ij}-b_i.
\]
In the \textsl{concave sink formulation},
we say that the pseudoflow $f$ is {\sl feasible}, if $e_i\ge 0$ for all $i\in V-t$ and $e_t>-\infty$. The objective is to maximize  $e_t$ for feasible pseudoflows.

Shigeno \cite{Shigeno06} defines this problem with
$e_i=0$ if $i\in V-t$, and $b\equiv 0$ and without explicit capacity constraints. She also discusses the version
with $e_i\ge 0$, and gives a reduction from the original version to this one. Whereas capacity constraints can be simulated by the functions $\Gamma_e$, we impose them explicitly as they will be included in the running time bounds. 
The formulation with $e_i\ge 0$ seems more natural as it gives a convex optimization problem, which is not the case for $e_i=0$.

\medskip

In the sink formulation, the node $t$ plays a distinguished role. It turns out to be more convenient  to handle all nodes equally. For this reason, we introduce another, seemingly more general version, called  the  \textsl{symmetric formulation} of both problems. Ideally, we would like to find a pseudoflow satisfying $e_i\ge 0$ for every $i\in V$. The formulation will be a relaxation of this feasibility problem, allowing
violation of the constraints, penalized by possibly different rates at different  nodes.

For each node $i\in V$ we are given a penalty factor $M_i>0$ and an auxiliary variable $\kappa_i\ge 0$. The objective is to minimize $\kappa_f=\sum_{i\in V} M_i\kappa_i$ for a pseudoflow $f$ subject to 
$e_i+\kappa_i\ge 0$ for each $i\in V$. 

The objective $\kappa_f$ is called the \textsl{excess discrepancy}. $\kappa_f=0$ means $e_i\ge 0$ for each $i\in V$. These conditions might be violated, but we have to pay penalty $M_i$ per unit  violation at $i$.

The sink version fits into this framework with $M_i=\infty$ for $i\neq t$ and $M_t=1$. 
However, it can be shown that setting finite, polynomially bounded $M_i$ values, the symmetric version returns an optimal (or sufficiently close approximate) solution to the  sink version, both for linear and for concave gain functions.
Besides the sink version, another natural setting is when $M_i=1$ for all $i\in V$, that is, maintaining $e_i\ge 0$ has the same importance for all nodes.

While the symmetric formulation could seem more general than the sink version, it can indeed be reduced to it. For an instance of the symmetric version with graph $G=(V,E)$, let us add a new sink node $t$ with an arc from $t$ to every node $i\in V$ with gain factor $1/M_i$.
Solving the sink version for this extended instance gives an optimal solution to the original problem.
The reason for introducing the symmetric formulation is its pertinence to our algorithmic purposes.}{
We define two closely related variants of the concave generalized flow problem. The first is essentially the problem
proposed by Truemper \cite{Truemper78} and Shigeno \cite{Shigeno06}.
Let $G=(V,E)$ be a directed graph. 
Let $n=|V|$, $m=|E|$, and for each node $i\in V$, let $d_i$ be the total number of incoming and outgoing arcs incident to $i$.

We are given lower and upper arc capacities $\ell,u:E\rightarrow \R$  and a monotone increasing concave function $\Gamma_{ij}:[\ell_{ij},u_{ij}]\rightarrow \R\cup\{-\infty\}$ on each arc and node demands
$b:V\rightarrow \R$.
By a {\sl pseudoflow} we mean a function  $f:E\rightarrow \R$ with $\ell\le f\le u$.
Given the pseudoflow $f$, let 
\[
e_i:=\sum_{j:ji\in E}\Gamma_{ji}(f_{ji})-\sum_{j:ij\in E}f_{ij}-b_i.
\]
In the first variant of the problem, called the \textsl{sink formulation},
there is a distinguished {\em sink} node $t\in V$.
The pseudoflow $f$ is {\sl feasible}, if $e_i\ge 0$ for all $i\in V-t$ and $e_t>-\infty$. The objective is to maximize  $e_t$ for feasible pseudoflows.

Shigeno \cite{Shigeno06} defines this problem with
$e_i=0$ if $i\in V-t$, and $b\equiv 0$ and without explicit capacity constraints. She also discusses the version
with $e_i\ge 0$, and gives a reduction from the original version to this one.
Whereas capacity constraints can be simulated by the functions $\Gamma_e$, we impose them explicitly as they will be included in the running time bounds. 
The formulation with $e_i\ge 0$ seems more natural as it gives a convex optimization problem, which is not the case for $e_i=0$.

\medskip

In the sink formulation, the node $t$ plays a distinguished role. It turns out to be more convenient to handle all nodes equally. For this reason, we introduce another, seemingly more general version, called  the  \textsl{symmetric formulation}. Ideally, we would like to find a pseudoflow satisfying $e_i\ge 0$ for every $i\in V$. The formulation will be a relaxation of this feasibility problem, allowing
violation of the constraints, penalized by possibly different rates at different  nodes.

For each node $i\in V$ we are given a penalty factor $M_i>0$ and an auxiliary variable $\kappa_i\ge 0$. The objective is to minimize $\kappa_f=\sum_{i\in V} M_i\kappa_i$ for a pseudoflow $f$ subject to 
$e_i+\kappa_i\ge 0$ for each $i\in V$. 
The objective $\kappa_f$ is called the \textsl{excess discrepancy}. $\kappa_f=0$  means $e_i\ge 0$ for each $i\in V$. These conditions might be violated, but we have to pay penalty $M_i$  per unit violation at $i$.

The sink version fits into this framework with $M_i=\infty$ for $i\neq t$ and $M_t=1$.
However, it can be shown that setting finite, polynomially bounded $M_i$ values, the symmetric version returns an optimal (or sufficiently close approximate) solution to the  sink version, both for linear and for concave gain functions.
While the symmetric formulation could seem more general than the sink version, it can indeed be reduced to it. For an instance of the symmetric version with graph $G=(V,E)$, let us add a new node $t$ with an arc from $t$ to every node $i\in V$ with gain factor $1/M_i$.
The reason for introducing the symmetric formulation is its pertinence to our algorithmic purposes.
}

\subsection{Complexity model}\label{sec:complexity}
\iftoggle{full}{The complexity setting will be different for generalized flows and concave generalized flows.
For generalized flows, we aim to find an optimal solution, while in the concave case, only an approximate
 one. For generalized flows, the gain functions are given explicitly as linear
functions, while in the concave case,}{From a complexity perspective,} the description of the functions might be infinite. To handle this difficulty, following the approach of Hochbaum and Shantikumar \cite{Hochbaum90}, we assume oracle access to the $\Gamma_{ij}$'s: our running time
estimation will give a bound on the number of necessary oracle calls. 
Two kinds of oracles are needed: {\em(i)} value oracle, returning $\Gamma_{ij}(\alpha)$ for
any $\alpha\in [\ell_{ij},u_{ij}]$; and {\em(ii)} inverse value oracle, returning a value $\beta$ with
$\alpha=\Gamma_{ij}(\beta)$ for any  $\alpha\in [\Gamma_{ij}(\ell_{ij}),\Gamma_{ij}(u_{ij})]$.

 We assume that both oracles return the exact (possibly irrational) solution, and any oracle query is done in $O(1)$ time.
Also, we assume any basic arithmetic operation is performed in $O(1)$ time, regardless to size and representation of the possibly irrational numbers. We expect that our results naturally extend to the setting with only approximate oracles and computational capacities in a straightforward manner. Notice that in an approximate sense, an inverse value oracle can be simulated by a value oracle.

By an {\sl $\varepsilon$-approximate solution} to the symmetric concave generalized flow problem we mean a feasible solution with the excess discrepancy larger than the optimum by at most $\varepsilon$. An {\sl $\varepsilon$-approximate solution} to the sink version means a solution with the objective value $e_t$  at most $\varepsilon$ less than the optimum, 
and the total violation of the inequalities $e_i\ge 0$ for $i\in V-t$ is also at most $\varepsilon$. 

\iftoggle{full}{In both cases, we}{Let us} assume that all $M_i$ values are positive integers, and let $M$ denote their maximum.

\iftoggle{full}{
For generalized flows, we assume all $\ell,u$ and $b$ are given as integers and $\gamma$ as rational numbers; let $B$ be the largest integer used in their descriptions. The running time bound will be $O(m^2(m\log B+\log M)\log n)$
for the symmetric formulation and $O(m^2(m+n\log n)\log B)$ for the sink formulation. This is the same as 
the complexity bound of the highest gain augmenting path algorithm by Goldfarb, Jin and Orlin \cite{Goldfarb97}. The best current running time bounds are
$O(m^{1.5}n^2\log B)$ using an interior point approach by Kapoor an Vaidya \cite{Kapoor96}, and   $\tilde O(m^2n\log B)$ by
Radzik \cite{Radzik04}, that is an enhanced version of \cite{Goldfarb97}.

\medskip

For the concave setting, we allow irrational capacities as well; in}{In} the complexity estimation, we will have $U$ as an upper bound on the absolute values on the  $b_i$'s, the capacities $\ell_{ij},u_{ij}$ and the $\Gamma_{ij}(\ell_{ij})$,
$\Gamma_{ij}(u_{ij})$ values. 
For each arc $ij$, let us define $r_{ij}=|\Gamma_{ij}(\ell_{ij})|$ whenever
 $\Gamma_{ij}(\ell_{ij})>-\infty$ and $r_{ij}=0$ otherwise.
Let 
\[
U=\max\left\{\max\{|b_i|: i\in V\}, \max\{|\ell_{ij}|,|u_{ij}|,|\Gamma_{ij}(u_{ij})|,r_{ij}: ij\in E\}\right\}.
\]

For the sink version, we need to introduce one further complexity parameter $U^*$ due to  difficulties arising if   $\Gamma_{ij}(\ell_{ij})=-\infty$ for certain arcs.
Let $U^*$ satisfy $U\le U^*$, and that  $e_t\le U^*$ for any pseudoflow  (it is easy to see that $U^*=d_tU$ always satisfies this property).
We also require that whenever there exists a feasible solution to the problem (that is, $e_i\ge 0$  for each $i\in V-t$ and $e_t>-\infty$), there exists one with $e_t\ge -U^*$.
If $\Gamma_{jt}(\ell_{jt})>-\infty$ for each arc $jt\in E$, then
$U^*=d_tU$ also satisfies this property. However, we allow $-\infty$ values as in the market applications we also have logarithmic gain functions.
A bound on $U^*$ \iftoggle{full}{ can be given as in Section~\ref{sec:market}. }{ is given in the full version for such cases.}

\medskip

The main result  is  as follows:
\begin{theorem}
For the symmetric formulation of the concave generalized flow problem, there exists a combinatorial algorithm that finds an $\varepsilon$-approximate solution in  $O(m(m+n\log n)\log(MUm/\varepsilon))$. 
For the sink formulation, there exists a combinatorial algorithm that finds an $\varepsilon$-approximate solution in $O(m(m+n\log n)\log(U^*m/\varepsilon))$. In both cases, the running time bound is on the number of arithmetic operations and oracle queries.
\end{theorem}

\nottoggle{full}{
For linear generalized flows, we are interested in finding exact solutions and therefore we use a different complexity model from the one described above. We assume all $\ell,u$ and $b$ are given as integers and $\gamma$ as rational numbers; let $B$ be the largest integer used in their descriptions. We obtain a running time bound $O(m^2(m\log B+\log M)\log n)$
for the symmetric and $O(m^2(m+n\log n)\log B)$ for the sink formulation (see the full version). This is the same as 
the complexity bound of the highest gain augmenting path algorithm \cite{Goldfarb97}. The best current running time bounds are
$O(m^{1.5}n^2\log B)$ using an interior point approach  \cite{Kapoor96}, and   $\tilde O(m^2n\log B)$ \cite{Radzik04}, an enhanced version of \cite{Goldfarb97}.
}{}

\medskip

The starting point of our investigation is the \textsc{Fat-Path} algorithm \cite{Goldberg91}.
The first main contribution is the introduction of the symmetric formulation. This is a more flexible framework, and thus we will be able to entirely avoid cycle cancellation and use  excess transportation phases only. Our (linear) generalized flow algorithm \nottoggle{full}{(presented in the full version) }{}is the first generalized flow algorithm that uses a pure scaling technique, without any cycle cancellation. The key new idea here is the way `$\Delta$-positive' and `$\Delta$-negative' nodes are defined, maintaining a `security reserve' in each node that compensates for adjustments when moving from the $\Delta$-scaling phase to the $\Delta/2$-phase.

We extend the linear algorithm to the concave setting using a local linear approximation of the gain functions, following Shigeno \cite{Shigeno06}. This approximation  is motivated by  the technique of Minoux \cite{Minoux86} and Hochbaum and Shantikumar \cite{Hochbaum90} for minimum cost flows with separable convex objectives.

\section{Applications to market equilibrium \iftoggle{full}{and Nash-bargaining }{} problems}\label{sec:market-app}
Intensive research has been pursued in the last decade to develop polynomial-time combinatorial algorithms
for certain market equilibrium problems. The starting point is the algorithm for computing market clearing prices in Fisher's model with linear utilities by Devanur et al. \cite{Devanur08}, followed by a study of several variations and extensions of this model. For a survey, see
\cite[Chapter 5]{Nisan07} or \cite{Vazirani11}.

In the {\sl linear  Fisher market model}, we are given a set $B$ of buyers and a set $G$ of goods. Buyer $i$ has a budget $m_i$, and there is one divisible unit of each good to be sold.
 For each buyer $i\in B$ and good $j\in G$, $U_{ij}\ge 0$ is the utility accrued by buyer $i$ for one unit of good $j$.
Let $n=|B|+|G|$ and $m$ be the number of pairs $ij$ with $U_{ij}>0$. Let $U_{\max}=\max\{U_{ij}: i\in B, j\in G\}$ and
$R=\max\{m_i: i\in B\}$. 
An equilibrium solution consist of prices $p_i$ on the goods and an allocation $x_{ij}$, so that 
{\em (i)} all goods are sold, {\em (ii)} all money of the buyers is spent, and {\em (iii)} each buyers $i$ buys a best bundle of goods, that is, goods $j$ maximizing $U_{ij}/p_j$. 

The equilibrium solutions for linear  Fisher markets were described via a convex program by Eisenberg and Gale \cite{Eisenberg59} in 1959; the combinatorial algorithms for this problem and other models rely on the KKT-conditions for the corresponding convex programs. Exact optimal solutions can be found, since these problems admit rational optimal solutions. 
\begin{align}
\max& \sum_{i\in B}m_i\log z_i\notag\\
z_i&\le \sum_{j\in G}U_{ij}x_{ij} \quad\forall i\in B\tag{EG}\label{eq:EG}\\
\sum_{i\in B}x_{ij}&\le 1 \quad\forall j\in G \notag\\
z,x&\ge 0\notag 
\end{align}
We show that the Eisenberg-Gale convex program, along with all extensions studied so far, falls into the broader class of convex generalized flows. Moreover, in all these extension we may replace linear or piecewise linear concave functions by arbitrary concave ones, still solvable approximately by our algorithm.

For the Eisenberg-Gale program, let us define the graph $(V,E)$ with $V=B\cup G\cup \{t\}$.
Let $ji\in E$ whenever $j\in G$, $i\in B$, $U_{ij}>0$, and set $\Gamma_{ji}(\alpha)=U_{ij}\alpha$ as a linear gain function. Also, let $it\in E$ for every $i\in B$ with $\Gamma_{it}(\alpha)=m_i\log \alpha$. Finally, set $b_j=-1$ 
for $j\in G$, and $b_i=0$ for $i\in B$.  The above program describes exactly the sink version of this concave generalized flow instance with $f_{ji}=x_{ij}$ for $i\in B$, $j\in G$ and $f_{it}=z_i$. (To formally fit into the model, we may add upper capacities $u_{ji}=1$ and $u_{it}=\sum_{j\in G}U_{ji}$ without changing the set of feasible solutions.)
Hence our general algorithm gives an $\varepsilon$-approximation for this problem. 
\iftoggle{full}{In Section~\ref{sec:market},}{In the full version,}{} we show that for sufficiently small $\varepsilon$ we can transform it to an exact optimal solution.

The flexibility of the concave generalized flow model enables various extensions. For example, we can replace each linear function $U_{ji}\alpha$ by an arbitrary concave increasing function, obtaining the perfect price discrimination model of Goel and Vazirani \cite{Goel10}. They studied  piecewise linear utility functions; our model enables arbitrary functions (although a rational optimal solution does not necessarily exist anymore).

In the \textsl{Arrow-Debreu Nash bargaining (ADNB)} defined by Vazirani \cite{Vazirani11}, traders arrive to the market with initial endowments of goods, giving utility $c_i$ for player $i$. They want to  redistribute the goods to obtain higher utilities using Nash bargaining.
The  disagreement point is when everyone keeps the initial endowment, guaranteeing her $c_i\ge 0$ utility. In an optimal Nash bargaining solution
we maximize $\sum_{i\in B}\log (z_i-c_i)$ over the constraint set in (\ref{eq:EG}). 
Unlike for the linear Fisher model, equilibrium prices may not exist, corresponding to a disagreement solution.
A sophisticated two phase algorithm is given in \cite{Vazirani11}, first for deciding feasibility, then for finding the equilibrium solution. 

The convex program for nonsymmetric ADNB can be obtained from the Eisenberg-Gale program by 
modifying the first set of inequalities to $z_i\le \sum_{j\in G}U_{ij}x_{ij}-c_i$. In the formulation as a concave generalized flow, this corresponds to modifying the $b_i=0$ values for $i\in B$ to $b_i=c_i$. Hence this problem also fits into our framework. From this general perspective, it does not  seem more difficult than the linear Fisher model.

Nonsymmetric Nash-bargaining was defined by Kalai \cite{Kalai77}. For ADNB, it corresponds to maximizing
 $\sum_{i\in B}m_i\log (z_i-c_i)$ over the constraint set in (\ref{eq:EG}), for some positive coefficients $m_i$.
 The algorithm in \cite{Vazirani11} heavily relies on the assumption $m_i=1$, and does not extend to this more general setting, called \textsl{nonsymmetric ADNB}. Finding a combinatorial algorithm for this latter problem was left open in \cite{Vazirani11}.
Another open question in \cite{Vazirani11} is to devise a combinatorial algorithm for (nonsymmetric) ADNB with piecewise linear, concave utility functions.  Our result generalizes even further, for arbitrary concave utility functions, since the linear functions $U_{ij}\alpha$ can be replaced by arbitrary concave functions.

Let $C=\max c_i$.
In \iftoggle{full}{Section~\ref{sec:market}}{the full version}, we show how our algorithm can be used to find an exact solution to the nonsymmetric ADNB problem in time $O(m(m+n\log n)(n\log (nU_{\max}R)+\log C))$.
The running time bound in \cite{Vazirani11} for symmetric ADNB  $(R=1)$
is $O(n^8\log U_{\max}+n^4\log C)$. 

\medskip

Let us also remark that an alternative convex program for the linear Fisher market, given by Shmyrev \cite{Shmyrev09}, shows that it also fits into the framework of minimum-cost circulations with a separable convex cost function, and thus can be solved by the algorithms of Hochbaum and Shantikumar \cite{Hochbaum90} or
Karzanov and McCormick \cite{Karzanov97}. 
Recently, \cite{Vegh11b} gave a strongly polynomial algorithm for a class of these problems, which includes Fisher's market with linear and spending constraint utilities. 
However, this does not seem to capture perfect price discrimination or ADNB, where no alternative formulations analogous to \cite{Shmyrev09} are known.

As further applications of the concave generalized flow model, we can take single-source multiple-sink markets by Jain and Vazirani \cite{Jain10}, or  concave cost matchings studied by Jain \cite{Jain11}.

\iftoggle{full}{
\medskip 
A distinct characteristic of the Eisenberg-Gale program and its extensions is that they are rational convex programs. We may loose this property when changing to general concave spending constraint utilities. However, for the case when the existence of a rational solution is guaranteed, one would prefer finding an exact optimal solution.
Section~\ref{sec:market} addresses the question of rationality. Theorem~\ref{thm:conv} shows that under certain technical conditions, our approximation algorithm can be turned into a polynomial time algorithm for finding an exact optimal solution. We shall verify these conditions for nonsymmetric ADNB.}
%

\iftoggle{full}{
\section{Background and overview}\label{sec:background}
The minimum-cost circulation problem\footnote{We shall use the term `circulation' to distinguish from other flow problems in the paper.} is fundamental to all problems and algorithms discussed in the paper. 
We give an overview  in Section~\ref{sec:circulation}. We present the two main algorithmic paradigms, cycle cancelling and successive shortest paths along with their efficient variants.
As already revealed by early studies of the problem (e.g. \cite{Onaga67,Truemper77}), there is a deep connection between generalized flows and classical minimum-cost circulations: the dual structures are quite similar, and the generalized flow algorithms stem from the classical algorithms for minimum-cost circulations.
 In Section~\ref{sec:overview-genflow}, we continue with an overview of generalized flow algorithms, exhibiting some important  ideas and their relation to minimum-cost circulations. 
We also exhibit here the main ideas of our algorithm for the linear case.
Section~\ref{sec:overview-convex} considers a different convex extension of minimum-cost circulations, when the linear cost function is replaced by a separable convex one. We show how the two main paradigms extend to this case, using different approximation strategies of the nonlinear functions. Finally in Section~\ref{sec:overview-concavegen} we consider the concave generalized flow problem, discuss the algorithm by Shigeno \cite{Shigeno06} and its relation to algorithmic ideas of the previous problems. We emphasize some difficulties and outline the ideas of our solution.

\subsection{Minimum-cost flows: cycle cancelling and successive shortest paths}\label{sec:circulation}
In the  minimum-cost circulation problem, given is  a directed graph $G=(V,E)$ with lower and upper arc capacities $\ell,u:E\rightarrow \R\cup\{\infty\}$, costs 
$c:E\rightarrow \R$ on the arcs and node demands $b:V\rightarrow \R$ with $\sum_{i\in V} b_i=0$.
Let 
\[
e_i=\sum_{j:ji\in E}f_{ji}-\sum_{j:ij\in E}f_{ij}-b_i.
\]
$f:E\rightarrow \R$ with $\ell\le f\le u$ is called a feasible circulation, if $e_i=0$ for all $i\in V$.
The objective is to minimize $c^Tf$ for feasible circulations.

Linear programming duality  provides the following characterization of  optimality.
For a feasible circulation $f$, let us define the residual graph $G_f=(V,E_f)$ with $ij\in E_f$ if
$ij\in E$ and $f_{ij}<u_{ij}$, or if $ji\in E$ and $\ell_{ji}<f_{ji}$. The first type of arcs are called {\sl forward arcs} and
are assigned the original cost $c_{ij}$, while the latter arcs are {\sl backward arcs} assigned cost $-c_{ji}$.
For notational convenience, we will use $f_{ij}=-f_{ji}$ on backward arcs.
Then $f$ is optimal if and only if $E_f$ contains no negative cost cycles. This is further equivalent to the existence of  a feasible potential $\pi:V\rightarrow \R$ with
$\pi_j-\pi_i\le c_{ij}$ for all arcs $ij\in E_f$.

Two main frameworks for minimum-cost flow algorithms are as follows. In the \textsl{cycle cancelling framework} (see e.g. \cite[Chapter 9.6]{amo}), we maintain a feasible circulation in each phase, with strictly increasing objective values. If the current solution is not optimal, the above conditions guarantee a negative cost cycle in the residual graph; such a cycle can be found efficiently. Sending some flow around this cycle decreases the objective and maintains feasibility, providing the next solution.

In the \textsl{successive shortest path framework} (see e.g. \cite[Chapter 9.7]{amo}), we waive feasibility by allowing $e_i>0$ or $e_i<0$; we call such nodes {\sl positive} and {\sl negative}, respectively. However,  we maintain dual optimality in the sense that the residual graph of the current pseudoflow contains no negative cost cycles in any iteration (or equivalently, admits a feasible potential). If there exists some positive and negative nodes, we send some flow from a positive node to a negative one using a minimum-cost path in the residual graph. This maintains dual optimality, and decreases the total $e_i$ value of positive nodes.

For rational input data, both these algorithms are finite, but may take an exponential number of steps (and might not even terminate for irrational input data). Nevertheless, using (explicit or implicit) scaling techniques, both can be implemented to run in polynomial time, and even in strongly polynomial time. 

A strongly polynomial version of the cycle cancellation algorithm is due to Goldberg and Tarjan \cite{Goldberg89}. In each step, a minimum mean cycle is chosen. In dual terms, we relax primal-dual optimality conditions  to $\pi_j-\pi_i\le c_{ij}+\varepsilon$ for  $ij\in E_f$, with $\varepsilon$ being equal to the negative of the minimum mean cycle value, decreasing exponentially over time.

Polynomial implementations of the successive shortest path algorithm can be obtained by capacity scaling; the most efficient, strongly polynomial such algorithm is due to Orlin \cite{Orlin93}.
We describe here a basic capacity scaling framework by Edmonds and Karp \cite{Edmonds72}. Instead of the residual graph $E_f$, we consider the $\Delta$-residual graph $E_f^\Delta$ consisting of arcs with residual capacity at least $\Delta$ (the residual capacity is  $u_{ij}-f_{ij}$ on a forward arc $ij$ and  $f_{ji}-\ell_{ji}$ on a backward arc). The algorithm consists of $\Delta$-scaling phases, with $\Delta$ decreasing by a factor of 2 between two phases.
In a $\Delta$-phase, we iteratively send $\Delta$ units of flow from a positive node $s$ with $e_s\ge \Delta$ to a negative node $t$ with
$e_t\le -\Delta$ on a minimum-cost path in $E_f^\Delta$. 
The $\Delta$-phase finishes when this is no longer possible, which means the total positive excess is at most $n\Delta$.

In the $\Delta$-phase, $\pi_j-\pi_i\le c_{ij}$ is maintained on arcs of the $\Delta$-residual graph.
When moving to the $\Delta/2$ phase, this might not hold anymore, since 
the $\Delta/2$-residual graph contains more arcs,  namely the ones with residual capacity between $\Delta/2$ and $\Delta$. At the beginning of the next phase, we 
saturate all these arcs, thereby increasing the positive excess to at most $(2n+m)\Delta/2$. This guarantees that the next phase will consist of at most $(2n+m)$ path
augmentations.

\subsection{Linear generalized flows -- cycle cancelling and excess transportation}\label{sec:overview-genflow}
In what follows, we consider the sink version of the generalized flow problem, with sink $t\in V$.
For a pseudoflow $f:E\rightarrow \R$, let us define the residual network $G_f=(V,E_f)$ as for circulations, with gain factor $\gamma_{ij}=1/\gamma_{ji}$ on backward arcs.
Consider a cycle $C$ in $E_f$. We can modify $f$ by sending some flow around $C$ from some $i\in V(C)$.
This leaves $e_{j}$ unchanged if $j\neq i$, and increases $e_{i}$ by $(\gamma(C)-1)\alpha$, where $\gamma(C)=\Pi_{e\in C}\gamma_e$. If $\gamma(C)>1$ then we call $C$ a \textsl{ flow-generating cycle}, while
for $\gamma(C)<1$, a \textsl{ flow-absorbing cycle}, since we can generate or eliminate excess at an arbitrary node $i\in C$, respectively. The amount of flow that can be generated is of course bounded by the capacity constraints.

To augment the excess of the sink $t$, we have to send the excess generated at a flow-generating cycle $C$ to  $t$. Hence we call a pair $(C,P)$ a \textsl{ generalized augmenting path (GAP)}, if {\em(a)}  $C$ is a flow-generating cycle, $i\in V(C)$, and $P$ is a path in $E_f$ from $i$ to $t$; or {\em (b)} $C=\emptyset$, and $P$ is a path in $E_f$ from some node $i$ with $e_{i}>0$ to $t$. 
Clearly, an optimal solution $f$ may admit no GAPs.
This is indeed an equivalence: $f$ is optimal if and only if no GAP exists.

The gain factors $\gamma_e$ play a role analogous to the costs $c_e$ for minimum-cost circulations. Indeed,
$C$ is a flow generating cycle if and only if it is a negative cost cycle for the cost function $c_e=-\log \gamma_e$.
The dual structure for generalized flows is also analogous to potentials. 
Let us call $\mu:V\rightarrow \R_{>0}\cup\{\infty\}$ with with $\mu_t=1$  a label function. 
\textsl{Relabeling} the pseudoflow $f$ by $\mu$ means dividing the flow on each arc $ij$ going out from $i$ by $\mu_i$. We get a problem equivalent to the original by replacing each arc gain by $\gamma^\mu_{ij}=\gamma_{ij} \mu_i/\mu_j$.
The labeling is called \textsl{conservative} if $\gamma_{ij}^\mu\le 1$ for all $ij\in E_f$, that is, no arc may increase the relabeled flow.
%
%

Assume we have a conservative labeling $\mu$ so that $e_i=0$ whenever $i\in V-t$, $\mu_i<\infty$.
 Let $V'\subseteq V$ denote the set of nodes from which there exists a directed path to $t$. It follows that {\em(i)} $\mu_i<\infty$ for all $i\in V'$, and {\em(ii)} $V'$ contains no flow-generating cycles. Consequently, given a conservative labeling, no GAP can exist, and the converse can also be shown to hold.
Note that on $V'$, $\pi_i=-\log \mu_i$ is a feasible potential for $c_e=-\log \gamma_e$ if and only if $\mu$ is conservative.

Based on this correspondence, minimum-cost circulation algorithms can be directly applied for generalized flows as a subroutine for eliminating all flow-generating cycles. This can be indeed implemented in strongly polynomial time, see \cite{Radzik93,Shigeno07}. The novel difficulty for generalized flows is how  to transport the 
generated excess from various nodes of the graph to the sink $t$. In the algorithm of Onaga \cite{Onaga67},
flow is transported iteratively on highest gain augmenting paths, that is, from $i\in V$ with $e_i>0$ on an $i-t$ path $P$ that maximizes $\gamma(P)=\Pi_{e\in P}\gamma_e$. It can be shown that using such paths does not
create any new flow generating cycles. Thus after having eliminated all type {\em(a)} GAPs, we only have to take care of type {\em(b)}. Unfortunately, this algorithm may run in exponential time (or may not even terminate for irrational inputs). This is due to the analogy between Onaga's algorithm and the successive shortest path algorithm -- observe that a highest gain path is a minimum-cost path for $-\log \gamma_e$. 

The first algorithms to overcome this difficulty and thus establish polynomial running time bounds were the two given by  Goldberg, Plotkin and Tardos \cite{Goldberg91}. One of them,   \textsc{Fat-Path}, uses a method analogous to capacity scaling. A path $P$ in $E_f$ from a node $i$ to $t$ is called $\Delta$-fat,
if assuming unlimited excess at $i$, it is possible to send enough flow along $P$ from $i$ to $t$ so that $e_t$ increases by $\Delta$. 

The algorithm consists of $\Delta$-phases, with $\Delta$ decreasing by a factor of 2 for the next phase. In the $\Delta$-phase, we first cancel all flow generating cycles. Then, from nodes $i$ with $e_i>0$, we transport flow on highest gain ones among the $\Delta$-fat paths. This might create new flow-generating cycles to be cancelled in the next phase. Nevertheless, it can be shown that at the beginning of a $\Delta$-phase, $e^*_t-e_t\le 2(n+m)\Delta$ for the optimum value $e^*_t$ and thus the number of path augmentations in each $\Delta$-phase can be bounded by $2(n+m)$.
Arriving at a sufficiently small value of $\Delta$, it is possible to obtain an optimal solution by a single maximum flow computation.

The basic framework of \cite{Onaga67} and of \textsc{Fat-Path}, namely using different subroutines for eliminating flow-generating cycles and for transporting excess to the sink has been adopted by most subsequent algorithms,
e.g. \cite{Goldfarb96,Goldfarb97,Tardos98,Fleischer02,Radzik04}. Among them, \cite{Goldfarb97} is an almost purely scaling polynomial time algorithm, 
but it still needs the an initial cycle-cancelling as in \cite{Onaga67}.

\medskip

In contrast, our algorithm does not need any cycle-cancelling, and adapts \textsc{Fat-Path} to a pure successive shortest paths framework.
The successive shortest paths algorithms for minimum-cost circulations start with an infeasible pseudoflow, having both positive and negative nodes. To use an analogous method for generalized flows, we have to give up the standard framework of algorithms where $e_i\ge 0$ is always maintained for all $i\in V-t$. This is the reason why we use the more flexible symmetric model: we start with possibly several nodes having $e_i<0$, and our aim is to eliminate them. An important property of the algorithm is that we always have to maintain $\mu_i=1/M_i$ for $e_i<0$; for this reason we shall avoid creating new negative nodes.

Similarly to \textsc{Fat-Path}, we use a scaling algorithm.
In the $\Delta$-phase, we consider the residual graph restricted to $\Delta$-fat arcs, arcs that may participate in a highest gain $\Delta$-fat-path, and maintain a conservative labeling $\mu$ with $\gamma_{ij}^\mu\le 1$ on the $\Delta$-fat arcs. When moving to the $\Delta/2$-phase, this condition may get violated due to $\Delta/2$-fat arcs that were not $\Delta$-fat.
Analogously to the Edmonds-Karp algorithm, we  modify the flow by saturating each violated arc and thereby restitute dual feasibility. However, these changes may create new negative nodes and thus violate the condition $\mu_i=1/M_i$ for $e_i<0$ we must  maintain.


We resolve this difficulty by maintaing a `security reserve' of $d_i\Delta\mu_i$ in each node $i$ ($d_i$ is the number of incident arcs). This gives an upper bound on the total change caused by restoring feasibility of incident arcs in all subsequent phases. 
We call a node $\Delta$-positive  if $e_i>d_i\Delta\mu_i$, $\Delta$-negative if $e_i<d_i\Delta\mu_i$ and $\Delta$-neutral if $e_i=d_i\Delta\mu_i$.
$\Delta$-negative nodes may become negative ($e_i<0$) at a later phase, and therefore we maintain the stronger condition
$\mu_i=1/M_i$ for them. We send flow from $\Delta$-positive nodes to $\Delta$-negative and $\Delta$-neutral ones. Thereby we treat some nodes
with $e_i>0$ as sinks and increase their excess further; however, as $\Delta$ decreases, such nodes may gradually become sources.

For the sink version,   described in Section~\ref{sec:sink}, we perform this algorithm with $M_i=B^n+1$ if $i\neq t$ and $M_t=1$. We shall show that this returns an optimal solution. We remark that the highest gain path algorithm \cite{Goldfarb97} can also be modified to a purely scaling algorithm using the symmetric formulation,
that enables to start from an arbitrary non-feasible solution and thereby eliminate the initial cycle-cancelling phase.

\subsection{Minimum-cost circulations with separable convex costs}\label{sec:overview-convex}
A natural and well-studied nonlinear extension of minimum-cost circulations is replacing each arc cost $c_e$ by
a convex function $C_e$. 
We are given a directed graph $G=(V,E)$ with  lower and upper arc capacities $\ell,u:E\rightarrow \R$, convex cost functions
$C_e:[\ell_e,u_e]\rightarrow \R$ on the arcs, and node demands $b:V\rightarrow \R$ with $\sum_{i\in V} b_i=0$.
Our aim is to minimize $\sum_{e\in E} C_e(f_e)$ for feasible circulations $f$. This is a widely applicable  framework, see \cite[Chapter 14]{amo}.

This is a convex optimization problem, and optimality can be described by the KKT conditions.
Let $C_{e}^+(\alpha)$ denote the left derivative of $C_e$. As before, for a feasible circulation $f$ define the auxiliary graph $G_f=(V,E_f)$.
Let $C_{ij}$ denote the original function if $ij$ is  a forward arc and let $C_{ji}(\alpha)=C_{ij}(-\alpha)$ on backward arcs. 
%
%
 $f$ is optimal if and only if there exists no cycle $C$ in $E_f$ with $\sum_{e\in C} C_e^+(f_e)<0$.
In dual terms, $f$ is optimal if and only if there exists a potential $\pi:V\rightarrow \R$ with $\pi_j-\pi_i\le C^+_{ij}(f_{ij})$ for all $ij\in E_f$.

Both the minimum mean cycle cancellation and the capacity scaling algorithms can be naturally extended to this problem with polynomial (but not strongly polynomial) running time bounds. However, these two approaches relax the optimality conditions in fundamentally different ways. The results \cite{Karzanov97} and \cite{Hochbaum90} address much more general problems: minimizing convex objectives over polyhedra given by matrices with bounded subdeterminants. 

Cycle cancellation was adapted by Karzanov and McCormick \cite{Karzanov97}. The algorithm subsequently cancels cycles in $E_f$ with minimum mean value respect to the $C^+_{e}(f_e)$ values. The only difference is that the flow augmentation around such a cycle might be less than what residual capacities would enable, in order to maintain
\begin{equation}\label{eq:KMc}
\pi_j-\pi_i\le C^+_{ij}(f_{ij})+\varepsilon \quad \forall ij\in E_f
\end{equation}
for the current potential $\pi$ and scaling parameter $\varepsilon$.

For capacity scaling, Hochbaum and Shanthikumar \cite{Hochbaum90} developed the following framework based on previous work of Minoux \cite{Minoux86} (see also \cite[Chapter 14.5]{amo}).
The algorithm consists of $\Delta$-phases. In the $\Delta$-phase, each $C_e$ is linearized with granularity $\Delta$. 

Let $E_f^\Delta$ denote the $\Delta$-residual network.
We will maintain $\Delta$-optimality, that is, there exists a potential $\pi$ such that
\begin{equation}\label{eq:Minoux}
\pi_j-\pi_i\le \frac{C_{ij}(f_{ij}+\Delta)-C_{ij}(f_{ij})}\Delta  \quad \forall ij\in E_f^\Delta.
%
\end{equation}
Let $\theta_\Delta({ij})$ denote the quantity on the right hand side.
If we increase flow on some arc $ij$ by $\Delta$ on some $ij$ for which equality holds, the resulting
pseudoflow will remain $\Delta$-optimal. We will always send $\Delta$ units of flows from a node $s$ with $e_s>0$ to a node $t$ with $e_t<0$ on a minimum-cost path in $E_f^\Delta$ with respect to  $\theta_\Delta({ij})$.
By the above observation, this maintains $\Delta$-optimality.

When moving to the next scaling phase replacing $\Delta$ by $\Delta/2$, we change to a better linear approximation of the $C_e$'s. Therefore, (\ref{eq:Minoux}) may get violated not only because $E_f^{\Delta/2}$ contains more arcs than $E_f^\Delta$ does, but also
on arcs already included in $E_f^\Delta$.
Yet it turns out  that modifying each $f_{ij}$ value by at most $\Delta/2$, (\ref{eq:Minoux}) can be re-established. This creates new (positive and negative) excesses of total at most $m\Delta$.

Recently, \cite{Vegh11b}  gave a strongly polynomial capacity scaling algorithm for a class of objective functions, including convex quadratic objectives, and Fisher's market with linear and with spending constraint utilities (based on Shmyrev's formulation).
Let us also remark that the results \cite{Karzanov97} and \cite{Hochbaum90} actually address much more general problems: minimizing convex objectives over polyhedra given by matrices with bounded subdeterminants. The framework of \cite{Hochbaum90} needs weaker assumptions on the objective function and on the oracle.

\subsection{Concave generalized flows}\label{sec:overview-concavegen}
As we have seen, both the cycle cancelling and capacity scaling approaches for minimum-cost circulations naturally extend to separable concave cost functions.
Similarly, our algorithm in Section~\ref{sec:concave} for concave gain functions is a natural
extension of the generalized flow algorithm in Section~\ref{sec:gen-alg}.

Nevertheless, we were not able to extend any previous generalized flow algorithm for concave gains.
Shigeno's \cite{Shigeno06} approach was to extend the \textsc{Fat-Path} algorithm of \cite{Goldberg91}. 
However, \cite{Shigeno06}  obtains 
polynomial running time bounds only for restricted classes of gain functions. The algorithm consists of two procedures applied alternately similar to \textsc{Fat-Path}: a cycle cancellation phase to generate excess on cycles with positive gains, and a path augmentation phase to transport new excess to the sink in chunks of $\Delta$. For both phases, previous methods naturally extend:
cycle cancelling is performed analogously to \cite{Karzanov97}, whereas 
path augmentation to \cite{Hochbaum90}. Unfortunately, fitting the two different methods together is problematic and does not yield polynomial running time.

The main reason is that the two approaches rely on fundamentally different kinds of approximation of the nonlinear gain functions. While for generalized flows, a cycle cancelling phase completely eliminates flow generating cycles, here we can only get an approximate solution allowing some small positive gain on cycles at termination.
In terms of residual arcs, we terminate with a condition analogous to (\ref{eq:KMc}) for concave cost flows. However, the path augmentation phase needs a linearization of the gain functions analogous to (\ref{eq:Minoux}). Notice that for $\varepsilon=0$, 
(\ref{eq:KMc}) implies (\ref{eq:Minoux}) for arbitrary $\Delta$. Yet if some small error $\varepsilon>0$ is allowed, then no general guarantee can be given so that (\ref{eq:Minoux}) hold for a certain value  $\Delta$. 

For this reason, our goal was to avoid using the two different frameworks simultaneously.
%
It turns out that our scaling based linear generalized flow algorithm (outlined in Section~\ref{sec:overview-genflow}) smoothly extends to this general setting. We use the local linearization 
 $\theta_{\Delta}^\mu(ij)$ of $\Gamma_{ij}$ used by Shigeno, an analogue of (\ref{eq:Minoux}). 
 In the $\Delta$-phase, we consider the graph of $\Delta$-fat arcs, and maintain $\theta_{\Delta}^\mu(ij)\le 1$ on them.

When moving from a $\Delta$-phase to a $\Delta/2$-phase in the linear algorithm, the only reason for infeasibility is due to  $\Delta/2$-fat arcs that were not $\Delta$-fat. In contrast, feasibility can be violated on $\Delta$-fat arcs as well, as  $\theta_{\Delta}^\mu(ij)\le 1<\theta_{\Delta/2}^\mu(ij)$ may happen
due to the finer linear approximation of the gain functions in the $\Delta/2$-phase. Fortunately, feasibility can be restored in this case as well, by changing the flow on each arc by a small amount.

\section{Linear generalized flow algorithm}\label{sec:gen-alg}
In this section, we investigate the symmetric formulation of the generalized flow problem.
In describing the optimality conditions, we also allow infinite $M_i$ values to 
incorporate the sink version.
However, in the algorithmic parts, we restrict ourselves to finite $M_i$ values.

We describe optimality conditions in Section~\ref{sec:genflow-opt}.
 Notion and results here are well-known in the generalized flow literature, thus we do not include references. Section~\ref{sec:cons-label} introduces $\Delta$-fat arcs and $\Delta$-conservative labelings, the feasibility framework in the $\Delta$-phase.
Section~\ref{sec:delta-canon} 
 describes the subroutine \textsc{Tighten-label}, an adaptation of  Dijkstra's algorithm that finds highest gain $\Delta$-fat paths. The main algorithm is exhibited in Section~\ref{sec:genalg-descr}.  Analysis and running time bounds are given in Section~\ref{sec:genflow-analysis}.
The final step of the algorithm is deferred to Section~\ref{sec:move-to-opt}, where  we show that when the total relabeled excess is sufficiently small, an optimal solution can be found by a single maximum flow computation.
\subsection{Optimality conditions}\label{sec:genflow-opt}
Let $G=(V,E)$ be a network with lower and upper capacities $\ell,u$ and node demands $b$.
In the sequel, we assume all lower capacities are 0. Every problem instance of symmetric generalized flows can be simply transformed to an equivalent one in this form in the following way. For each arc $ij\in E$,
increase the node demand $b_i$ by $\ell_{ij}$ and decrease $b_j$ by $-\gamma_{ij}\ell_{ij}$. Modify the lower capacity of $ij$ to 0 and the upper to 
$u_{ij}-\ell_{ij}$.

For a pseudoflow $f$, we define the residual network  $G_f=(V,E_f)$ as follows.
Let $ij\in E_f$, if $ij\in E$ and $f_{ij}<u_{ij}$ or if $ji\in E$ and $f_{ji}>0$.
The first type of arcs are called \textsl{forward arcs}, while the second type are the 
\textsl{backward arcs}.  
For a forward arc $ij$, let $\gamma_{ij}$ be the same as in the original graph. For
a backward arc $ji$, let $\gamma_{ji}=1/\gamma_{ij}$.
Also, we define $f_{ji}=-\gamma_{ij}f_{ij}$ for every backward arc $ji\in E_f$.
For backward arcs, the  capacities are $\ell_{ji}=-\gamma_{ji}u_{ji}$ and $u_{ji}=0$.
By increasing (decreasing) $f_{ji}$ by $\alpha$, we mean decreasing (increasing) $f_{ij}$ by $\alpha/\gamma_{ij}$.

Let $P=i_0\ldots i_k$ be a walk in the auxiliary graph $E_f$. By sending $\alpha$ units of flow along $P$, we mean  increasing each $f_{i_hi_{h+1}}$ by $\alpha\Pi_{0\le t<h} \gamma_{i_ti_{t+1}}$.
We assume $\alpha$ is chosen small enough so that no capacity gets violated. 
 Note that this decreases $e_{i_0}$ by $\alpha$, increases $e_{i_k}$ by 
$\alpha\Pi_{0\le t<k} \gamma_{i_ti_{t+1}}$, and leaves the other $e_i$ values unchanged. 

Let $C=i_0\ldots i_{k-1}i_0$ be a cycle in $E_f$.
Sending $\alpha$ units of flow on $C$ from $i_0$ modifies only $e_{i_0}$,
increasing by $(\gamma(C)-1)\alpha$  for $\gamma(C)=\Pi_{e\in C} \gamma_{e}$.
$C$  is called a \textsl{flow generating cycle} if $\gamma(C)>1$. On such a cycle
for any choice of $i_0\in V(C)$, we can create an excess of  $(\gamma(C)-1)\alpha$  by sending $\alpha$ units around  (assuming that $\alpha$ is sufficiently small so that no capacity constraints are violated).

The pair $(C,P)$ is called a \textsl{generalized augmenting path (GAP)} in the following cases:
\begin{enumerate}[(a)]
\item $C$ is a flow generating cycle, $i_0\in V(C)$, $t\in V$ is a node with $e_t<0$, and
 $P$ is a path in $E_f$ from $i_0$ to $t$ ($i_0=t$, $P=\emptyset$ is possible);
\item  
$C=\emptyset$, and
$P$ is a path between two nodes $s$ and $t$ with $e_s>0$, $e_t<0$;
\item $C=\emptyset$, and
$P$ is a path between  $s$ and $t$ with $e_s\le 0$, $e_t<0$ 
and $\gamma(P)=\Pi_{e\in C} \gamma_{e}>M_s/M_t$.
\end{enumerate}

\begin{lemma}\label{lem:opt-no-GAP}
If $f$ is an optimal solution, then no GAP exists.
\end{lemma}
\begin{proof}
In case {\em(a)}, we can send some $\alpha>0$ units of flow around $C$ from $i_0$, and then send the generated $(\gamma(C)-1)\alpha$ excess from $i_0$ to $t$ along $P$. For sufficiently small positive value of $\alpha$,
this is possible without violating the capacity constraints and it decreases the excess discrepancy.
In case {\em(b)}, we can decrease the excess discrepancy at $t$ while only decreasing a positive excess at $s$.
In case {\em(c)}, although   $M_s\kappa_{s}$ increases, $M_t\kappa_t$ decreases by a larger amount.
\end{proof}

\medskip
The dual description of an optimal solution is in terms of \textsl{relabelings} with a
label function $\mu: V\rightarrow \R_{>0}\cup\{\infty\}$.
For each node $i\in V$, let us  rescale the flow on all arcs $ij\in E$ by $\mu_i$: let $f^\mu_{ij}=f_{ij}/\mu_i$. We
get a problem equivalent to the original one with relabeled gains
 $\gamma_{ij}^\mu=\gamma_{ij}\mu_i/\mu_j$.
Accordingly, the relabeled demands, excesses, and capacities are
$b^\mu_i=b_i/\mu_i$, $e^\mu_i=e_i/\mu_i$, and $u_{ij}^\mu=u_{ij}/\mu_i$.
A relabeling is \textsl{conservative}, if for any residual arc $ij\in E_f$,
 $\gamma^{\mu}_{ij}\le 1$, that is, no arc may increase the relabeled flow.
Furthermore, for each $i\in V$, $\mu_i\ge 1/M_i$ is required and equality must hold whenever $e_i<0$. 

We use the conventions $\infty\cdot 0=0$ and $\infty/\infty=0$. Accordingly, if $\mu_i=\infty$, we define $b^\mu_i=e^\mu_i=0$, and 
$\gamma^{\mu}_{ji}=0$ for all arcs $ji\in E_f$. If $ij\in E_f$, $\mu_i=\infty$ and $\mu_j<\infty$,
then $\gamma^{\mu}_{ij}=\infty$.
Consequently, if $\mu$ is conservative, then $\mu_i<\infty$ for any $i\in V$ for which  there exists a path in $E_f$ from $i$ to any node $t\in V$ with $e_t<0$.
Also, if $\mu$ is conservative, there exists no flow generating cycle on the node set $\{i: \mu_i<\infty\}$. This is since for a cycle $C$,
$\gamma(C)=\Pi_{ij\in C} \gamma_{ij}=\Pi_{ij\in C}\gamma^\mu_{ij}$.

\begin{theorem}\label{thm:genflow-opt}
For a pseudoflow $f$, the following are equivalent.
\begin{enumerate}[(i)]
\item $f$ is an optimal solution to the symmetric generalized flow problem.
\item $E_f$ contains no generalized augmenting paths.
\item There exists a conservative relabeling $\mu$ with  $e_i=0$ whenever $1/M_i<\mu_i<\infty$.
\end{enumerate}
\end{theorem}
\begin{proof}
The equivalence of {\em(i)} and {\em(iii)} is by linear programming duality, with $\mu_i$ being the reciprocal
of the dual variable corresponding to the inequality $e_i+\kappa_i\ge 0$. 
{\em(i)} implies {\em(ii)} by Lemma~\ref{lem:opt-no-GAP}.

It is left to show that {\em(ii)} implies {\em(iii)}. If the excess discrepancy is 0 (that is, $e_i\ge 0$ for all $i\in V$), then
$\mu\equiv \infty$ is conservative. Otherwise, let $N=\{t: e_t<0\}$. If $E_f$ contains no directed path from $i\in V$ to $N$, then let $\mu_i=\infty$. For the other nodes $i\in V$, let $\mu_i$ be the smallest possible value of $1/(\gamma(P)M_t)$ for $\gamma(P)=\Pi_{e\in P}\gamma_e$, where $P$ is  a walk in $E_f$ starting from $i$ and ending in a node  $t\in N$. By {\em(ii)}, this is well-defined, since all cycles can be removed from a walk $P$ without decreasing $\gamma(P)$.

$\mu$ clearly satisfies $\gamma_{ij}\mu_i/\mu_j\le 1$. We shall prove $\mu_i\ge 1/M_i$ for each $i\in V$ and $\mu_t=1/M_t$ for each $t\in N$. 
If $e_i>0$, then no such path $P$ may exists as it would give a type {\em(b)} GAP. Consequently, $e_i=\infty$. If $e_i\le 0$, then $\mu_i\ge 1/M_i$ as otherwise the optimal $P$ defining $\mu_i$ would be a type {\em(c)} GAP. Finally, if $t\in N$, then $\mu_t\le 1/M_t$ as the gain of the path $P=\emptyset$ is defined as 1.
%
\end{proof}


\subsection{$\Delta$-conservative labels}\label{sec:cons-label}

The residual capacity of arc $ij\in E_f$ is $u_{ij}-f_{ij}$ (for a backward arc $ij$, this is $\gamma_{ji}f_{ji}$).
In contrast, we define the {\em \st} of $ij\in E_f$ by $s_f(ij)=\gamma_{ij}(u_{ij}-f_{ij})$ (on backward arcs, $s_f(ij)=f_{ji}$).
 The \st expresses the maximum possible flow increase in $j$ if we saturate $ij$.
This notion enables us to identify arcs that can participite in fat paths during the algorithm.
In accordance with the other variables, the relabeled \st is defined as $s_f^\mu(ij)=s_f(ij)/\mu_j$.

For the scaling parameter $\Delta>0$, we define the following relaxation of  conservativity.
Let $\mu:V\rightarrow \R_{>0}\cup\{\infty\}$ be a label function.
Recall that $d_i$ is the total number of arcs incident to $i$.
 A node $i\in V$ is called \textsl{$\Delta$-negative} if $e^\mu_i<d_i\Delta$, \textsl{$\Delta$-neutral} if $e^\mu_i=d_i\Delta$ and
\textsl{$\Delta$-positive} if $e^\mu_i>d_i\Delta$.

The {\em $\Delta$-fat graph} $E^\mu_f(\Delta)$ is the set of residual  arcs of relabeled \st at least $\Delta$: 
\[
E^\mu_f(\Delta)=\{ij: ij\in E_f: s_f^\mu(ij)\ge \Delta\}.
\]
Arcs in $E^\mu_f(\Delta)$ will be called {\em $\Delta$-fat arcs}. The labeling $\mu$ is called
 {\em $\Delta$-conservative}, if $\gamma_{ij}^\mu\le 1$ holds for every $ij\in E^\mu_f(\Delta)$, and 
$\mu_i\ge 1/M_i$ for all $i\in V$. Further, for every $\Delta$-negative node $i$, we require $\mu_i=1/M_i$.

Observe that no nodes with $\mu_i=\infty$ may be present in a $\Delta$-conservative labeling.
If $\Delta\le \Delta'$, then a $\Delta$-conservative labeling is also $\Delta'$-conservative, and $0$-conservativity is identical
to conservativity.
Let $Ex^\mu(f)=\sum_{i\in V} \max\{e_i^\mu,0\}$  and $Ex_\Delta^\mu(f)=\sum_{i\in V} \max\{e_i^\mu-d_i\Delta,0\}$ denote the total relabeled excess and total modified relabeled excess for $\Delta$, respecitively. Note that $Ex^\mu(f)\le Ex_\Delta^\mu(f)+2m\Delta$.

\begin{lemma}\label{lem:modify-Delta}
Let $f$ be a pseudoflow with a $\Delta$-conservative labeling $\mu$.
Let $0\le \Delta'<\Delta$. 
Then there exists a 
flow $\bar f$
such that $\bar f$ and $\mu$ are $\Delta'$-conservative
and $Ex_{\Delta'}^\mu(\bar f)\le Ex_{\Delta}^\mu(f)+3m(\Delta-\Delta')$
\end{lemma}
\begin{proof}
We shall construct $\bar f$ by modifying $f$ on each arc independently. For $\Delta$-fat arcs, $\Delta$-conservativity guarantees
$\gamma_{ij}^\mu\le 1$. Consider an arc $ij\in E_f^\mu(\Delta')-E_f^\mu(\Delta)$, that is, 
\[
\Delta'\le \gamma^\mu_{ij}(u^\mu_{ij}-f^\mu_{ij})<\Delta
\]
and assume  $\gamma_{ij}^\mu>1$.
Let us set $\bar f_{ij}=\min\{u_{ij},f_{ij}+\frac{(\Delta-\Delta')\mu_j}{\gamma_{ij}}\}$. Then $ij$ cannot be $\Delta'$-fat, since either $ij\notin E_f$, or $\gamma^\mu_{ij}(u^\mu_{ij}-\bar f^\mu_{ij})<\Delta'$.

$f^\mu_{ji}-\bar f^\mu_{ji}= \gamma_{ij}^\mu(\bar f^\mu_{ij}-f^\mu_{ij})\le \Delta-\Delta'$. Also,
$\bar f^\mu_{ij}-f^\mu_{ij}\le \Delta-\Delta'$ holds because of $\gamma_{ij}^\mu>1$.
We also have to consider the possibility that  $ji$ is $\Delta'$-fat for $\bar f$. In this case, conservativity is guaranteed since
   $\gamma_{ji}^\mu=1/\gamma_{ij}^\mu<1$.

To complete the proof of $\Delta'$-conservativity, we show that $\bar f$ has no $\Delta'$-negative nodes with $\mu_i>1/M_i$. 
As we have seen, both $f_{ij}$ and $\gamma_{ij}f_{ij}$ change by at most $\Delta-\Delta'$. Consequently for every $i\in V$, the total possible change of the relabeled flow on arcs incident to $i$ is $d_i(\Delta-\Delta')$. A node is nonnegative for  $\Delta$ if $e_i^\mu\ge d_i\Delta$ and for $\Delta'$ if  $e_i^\mu\ge d_i\Delta'$. Therefore, a nonnegative node cannot become $\Delta'$-negative, proving the claim. 

Further, $Ex_{\Delta'}^\mu(f)\le Ex_{\Delta}^\mu(f)+\sum_{i\in V}d_i(\Delta-\Delta')$, and on each arc, at most $\Delta-\Delta'$ units of new excess is created.
This gives $Ex_{\Delta'}^\mu(\bar f)\le Ex_{\Delta}^\mu(f)+3m(\Delta-\Delta')$.
\end{proof}
The proof also gives a straightforward algorithm for finding such an $\bar f$. Let \textsc{Adjust}($\Delta,\Delta'$) denote this subroutine.
In particular, \textsc{Adjust}($\Delta,0$) finds an $\bar f$ for which $\mu$ is a conservative labeling. Further, if there are no $\Delta$-negative nodes for $f$ and $\mu$, then $\bar f$ is a 0-discrepancy optimal solution. 

\subsection{ $\Delta$-canonical labels}\label{sec:delta-canon}
An edge $ij$ is called \textsl{tight} if $\gamma_{ij}^\mu=1$, an a directed path is \textsl{tight} if it consists of tight arcs.
 Given a pseudoflow $f$, a conservative labeling $\mu$ is called \textsl{canonical}, if for each $i\in V$ with $\mu_i<\infty$, there exists a tight path in $E_f$ from $i$ to a negative node.
Analogously, for $\Delta>0$, a labeling is called \textsl{$\Delta$-canonical}, if it is $\Delta$-conservative, and for each $i\in V$ there exists a tight path in $E^\mu_f(\Delta)$ from $i$ to some $\Delta$-negative or $\Delta$-neutral node.
Such a path is a highest gain $\Delta$-fat path as in \cite{Goldberg91}.
 (Note that $\mu_i<\infty$ for every $i\in V$ for $\Delta$-canonical labelings, and also that paths are allowed to end in $\Delta$-neutral nodes, in contrast to canonical labelings.) By $0$-conservative labeling we mean a conservative one.

 Given a $\Delta$-conservative relabeling $\mu$ which is not canonical,  \textsc{Tighten-Label}$(f,\mu,\Delta)$ replaces $\mu$ by a $\Delta$-canonical labeling $\mu'$ with $\mu'_i\ge \mu_i$ for each $i\in V$. 

We first describe \textsc{Tighten-Label}$(f,\mu,0)$, when it is essentially a multiplicative interpretation if Dijkstra's algorithm.
Let $V'\subseteq V$ be the set of nodes $i$ with a directed path in $E_f$ from $i$ to a negative node.
For nodes in $V\setminus V'$, let us set $\mu_i=\infty$.
Let $S\subseteq V'$ be the set of nodes $i$ for which there exists a (possibly empty) tight path for the current $\mu$ to a negative node. In each step of the algorithm, $S$ will be extended by at least one element, and we terminate if $S=V'$, when the current relabeling is canonical.

If $V'\setminus S\neq \emptyset$, let us multiply $\mu_i$ for each $i\in V'\setminus S$ by $\alpha$ defined as
\[
\alpha=\min\left\{\frac{1}{\gamma^\mu_{ij}}: 
ij\in E_f, i\in V'\setminus S, j\in S\right\}.
\]
By the definition of $S$, $\alpha>1$, and after multiplying by $\alpha$, at least one arc $ij\in E_f$
with $i\in V'\setminus S$, $j\in S$ will become tight. Tight arcs inside $S$ also remain tight, hence $S$ is extended by at most one node. Also, the choice of $\alpha$ guarantees that $\mu$ remains conservative.

\medskip

Let us now describe \textsc{Tighten-Label}$(f,\mu,\Delta)$  for $\Delta>0$.  The main difference is that increasing $\mu_i$ may turn a $\Delta$-positive node into $\Delta$-neutral. We have to stop increasing $\mu_i$ at this point and add $i$ to $S$. (This is in accordance with the goal that we allow tight paths to $\Delta$-neutral nodes as well.)
In each phase of the algorithm, let $S\subseteq V$ denote the subset of nodes that have a (possibly empty) tight path in $E_f^\mu(\Delta)$ to a $\Delta$-negative or $\Delta$-neutral node. $S$ is initialized as the set of $\Delta$-negative and $\Delta$-neutral nodes, and is extended by at least one element per phase. The algorithm terminates once $S=V$.

In every phase, we multiply  $\mu_i$ for every $i\in V\setminus S$ by the same factor $\alpha>1$.
Consider an arc $ij\in E_f(\Delta)$ with $i\in V-S$, $j\in S$. This must satisfy $\gamma_{ij}^\mu<1$.
 Increasing $\mu_i$ increases $\gamma_{ij}^\mu$. Note that the \st $s_f^\mu(ij)=s_f(ij)/\mu_j$ is not changed as it is not dependent of $\mu_i$.
By definition, all nodes in $V-S$ are $\Delta$-positive.
When increasing $\mu_i$, $e_i^\mu$ decreases and therefore $i$ may become $\Delta$-neutral. At this point, we have to stop to avoid creating new $\Delta$-negative nodes.
Let us define
\begin{align}
\alpha=\min\left\{
\min 
\left\{ \frac{1}{\gamma^\mu_{ij}}:  ij\in E_f^\mu(\Delta)\right\},
\min\left\{\frac{e^\mu_i}{d_i\Delta}: i\in V-S\right\}
\right\}.\label{eq:alpha-def}
\end{align}
Clearly, $\alpha>1$, and after multiplying each $\mu_i$ by $\alpha$ for $i\in S$, either we obtain a new $\Delta$-neutral node in $V-S$, or 
at least one arc $ij\in E_f$
with $i\in V\setminus S$, $j\in S$ will become tight. Tight arcs inside $S$ also remain tight, and their capacities is unchanged, hence $S$ is extended by at most one node. Also, the choice of $\alpha$ guarantees that $\mu$ remains $\Delta$-conservative. Note that arcs with $j\in V\setminus S$ may disappear from $E_F^\mu(\Delta)$ as their \st decreases.

\subsection{Description of the algorithm}\label{sec:genalg-descr}

\begin{figure}[htb]
\begin{center}
\fbox{\parbox{\textwidth}{
\begin{tabbing}
xxxxx \= xxx \= xxx \= xxx \= xxxxxxxxxxx \= \kill
\> \textbf{Algorithm} \textsc{Symmetric Fat-Path}\\
\> \textbf{for} $i\in V$ \textbf{do} $\mu_i\leftarrow \frac{1}{M_i}$;\\
\>   \textbf{for} $ij\in E$ \textbf{do} $f_{ij}\leftarrow 0$; \\
\>  $\Delta\leftarrow MB^2+1$;\\
\> \textbf{while }$(\konstterm)\Delta\ge 1/B^{\konstexp}$  \textbf{do}\\
\> \> \textbf{do}\\
\> \>  \> \textsc{Tighten-Label}$(f,\mu,\Delta)$;\\
\> \>  \> $D\leftarrow \{i\in V:  e_i>(d_i+1)\Delta\}$;\\
\> \> \> $N_0\leftarrow \{i\in V: e_i\le d_i\Delta\}$;\\
\> \> \> \textbf{pick} $s\in D$, $t\in N$ connected by a tight path $P$;\\
\> \> \> send $\Delta$ units of flow along $P$;\\
\> \>  \textbf{while} $D\neq \emptyset$;\\
\> \> \textsc{Adjust}$(\Delta,\frac{\Delta}2)$;\\
\> \>  $\Delta\leftarrow \frac{\Delta}2$;\\
\>  \textsc{Adjust}$(\Delta,0)$;\\
\> \textsc{Tighten-Label}$(f,\mu,0)$;\\
\> \textbf{compute} a maximum flow from nodes $\{s\in V: e_s>0\}$\\
\> \> \> \> to nodes  $\{t\in V: e_t<0\}$ using tight arcs for $\mu$;\\
\> \textbf{return} optimal primal solution $f$ and optimal dual solution $\mu$ \\
\end{tabbing}
}}
\caption{The algorithm for symmetric linear generalized flows}\label{code:genflow}
\end{center}
\end{figure}

The algorithm is shown on  Figure~\ref{code:genflow}.
We start with $\mu_i=1/M_i$ for every node $i\in V$, $f\equiv 0$ and $\Delta=MB^2+1$.
Once $(\konstterm)\Delta< 1/B^{\konstexp}$, then an optimal solution can be found by a single maximum flow computation, as  shown in Section~\ref{sec:move-to-opt}. In this case we terminate.

The algorithm consists of $\Delta$-scaling phases. 
During the $\Delta$-phase, we  maintain a pseudoflow $f$ along with a $\Delta$-conservative labeling $\mu$. The $\mu_i$ values can only increase. Let $N_0$ denote the set of $\Delta$-negative and $\Delta$-neutral nodes, and
 $D$ the set of nodes  $i$ with ${e}_i^\mu\ge  (d_i+1) \Delta$. The $\Delta$-phase consists of a sequence of iterations until
$D$ becomes empty.
In every iteration of the algorithm, we update $\mu$ to  a $\Delta$-canonical labeling by calling \textsc{Tighten-Label}$(f,\mu,\Delta)$.
If $D\neq\emptyset$ still holds, we  pick an arbitrary $s\in D$, and send $\Delta$ units of flow from $s$  to some $\Delta$-negative or $\Delta$-neutral $t\in N_0$ on a tight  path $P$.
At the end of the $\Delta$-phase, we modify $f$ by \textsc{Adjust}$(\Delta,\Delta/2)$, and proceed to the $\Delta/2$-phase.

\subsection{Analysis}\label{sec:genflow-analysis}
\begin{claim}\label{cl:initial}
The initial $\mu$ is $\Delta$-conservative, and $\Delta$-conservativity is maintained during the entire $\Delta$-phase.
\end{claim}
\begin{proof}
At the beginning $s_f^\mu(ij)=\gamma_{ij} u_{ij}/\mu_j\le MB^2<\Delta$ for every $ij\in E$, and $E_f$ contains no backward arcs. Consequently,
$E_f^\mu(\Delta)=\emptyset$ and 
$\Delta$-conservativity trivially holds. Also, $\mu_i=1/M_i$ holds for every node $i$.
 \textsc{Tighten-Label}$(f,\mu,\Delta)$ clearly maintains
$\Delta$-conservativity. 
This is also maintained when sending flow, as we only use tight arcs.
At the end of the $\Delta$-phase, \textsc{Adjust}$(\Delta,\Delta/2)$ transforms $f$ to a $\Delta/2$-conservative pseudoflow.
\end{proof}

\begin{claim}\label{cl:ex-delta}
At the beginning of every $\Delta$-phase, $Ex_\Delta^\mu(f)\le (\konstlin)\Delta$.
\end{claim}
\begin{proof}
This holds by definition in the first phase. 
In the first phase, $Ex_\Delta^\mu(f)\le M(\sum_{i\in V}b_i)=MB<(\konstlin)\Delta$.
Once we finish all iterations in the $\Delta$-phase, $D=\emptyset$ implies
$Ex_\Delta^\mu(f)\le n\Delta$. In \textsc{Adjust}$(\Delta,\Delta/2)$, we increase the excess by at most 
$3m\Delta/2$ (Lemma~\ref{lem:modify-Delta}.)
 Hence at the beginning of the $\Delta/2$-phase, $Ex_{\Delta/2}^\mu(f)\le (\konstlin)\Delta/2$, proving the claim.
\end{proof}

\begin{lemma}\label{lem:push-bound}
A $\Delta$-phase consists of at most $\konstlin$ iterations.
\end{lemma}
\begin{proof}
Consider the potential function $\Psi=\sum_{i\in V}\lfloor \max\{{e}^\mu_i-(d_i+1)\Delta,0\}/\Delta\rfloor$. 
By Claim~\ref{cl:ex-delta}, $\Psi\le\konstlin$ holds at the beginning.  
In the relabeling steps, $\Psi$ may only decrease, and in every path augmentation, it decreases by exactly 1.
\end{proof}

\begin{theorem}
The algorithm runs in $O(m(m+n\log n)(m\log B+\log M))$ time.
\end{theorem}
\begin{proof}
$\Delta$ always decreases by a factor of 2, its initial value is $MB^2$ and we terminate if $Ex_\Delta^\mu(f)<1/B^{\konstexp}$. Hence the number of phases can be bounded by
$O(m\log B+\log  M)$.
The number of iterations is  $O(m)$ by Lemma~\ref{lem:push-bound}.
The running time of an iteration is dominated by the \textsc{Tighten-Label} step, that can be done in $O(m+n\log n)$ time following  Fredman and Tarjan's \cite{Fredman87} implementation of Dijkstra's algorithm.
\end{proof}

\subsection{Moving to an optimal solution}\label{sec:move-to-opt}
If $(\konstterm)\Delta<1/B^{\konstexp}$ at a certain iteration of the algorithm, then by Claim~\ref{cl:ex-delta}, 
$Ex^\mu(f)+3m\Delta<1/B^{\konstexp}$ holds.
Next, we  transform $f$ to  $\bar f$ by \textsc{Adjust}($\Delta,0$), so that $\mu$ is a conservative labeling for $\bar f$. By Lemma~\ref{lem:modify-Delta}, $Ex^\mu(\bar f)<1/B^{\konstexp}$ follows.
Then \textsc{Tighten-Label}$(\bar f,\mu,\Delta)$ transforms $\mu$  into a canonical labeling.
By the following lemma, a single maximum flow computation yields an optimal solution. This 
is the standard technique how most algorithms in the literature terminate.

\begin{lemma}\label{lem:move-to-opt}
Let $\mu$ be a canonical labeling for $f$, and 
let $\tilde G_f=(V,\tilde E_f)$ be the subgraph of $G_f$ consisting of tight arcs in $E_f$.
If $Ex^\mu(f)<1/B^{\konstexp}$, then a single maximum flow computation on $\tilde G$  from source set $P=\{s\in V: e_s>0, \mu_s<\infty\}$ to sink set $N=\{t\in V: e_t<0\}$ terminates with an optimal solution.
\end{lemma}
\begin{proof}
Consider the flow $f'$ resulting after the maximum flow computation. 
Since flow was sent only on tight arcs, $\mu$ is also conservative for $f'$.
If there are no more nodes $s$ with $e^\mu_s>0$, then  by Theorem~\ref{thm:genflow-opt}, $f'$ is optimal. Assume now $P'$, the set of such nodes for $f'$ is nonempty.

Let $S\subseteq V$ be the set of nodes reachable from $P'$ using tight residual arcs in $E_{f'}$. By optimality, $S$ contains no node with negative excess.  
If an arc $ij\in E$ leaves $S$ then either $ij$ is saturated, that is, $f'_{ij}=u_{ij}$, or $\gamma_{ij}<1$ and $f'_{ij}=0$. Similarly, if $ij\in  E$ enters $S$, then $f'_{ij}=0$ must hold.
Also, on all arcs $ij$ with $i,j\in S$, $f'_{ij}>0$ either $\gamma^\mu_{ij}=1$ or $f'_{ij}=u_{ij}$. Let $F_1$ denote the set of such arcs with $f'_{ij}<u_{ij}$ (and
$\gamma_{ij}^\mu=1$), and $F_2$ the set of those with $f'_{ij}=u_{ij}$. Therefore, 
\begin{align}
0&\le Ex^\mu(f')=\sum_{i\in S} e^\mu_i(f')=\sum_{i\in S} \left(\sum_{j:ji\in E}\gamma_{ji}^\mu {f'_{ji}}^\mu-\sum_{j:ij\in E} {f'_{ij}}^{\mu}-b_i^\mu\right)\notag\\
&=\sum_{ij\in F_1}({\gamma_{ji}^\mu f'_{ji}}^\mu- {f'_{ji}}^\mu)+\sum_{ij\in F_2}u_{ij}^\mu(\gamma_{ij}^\mu-1)-\sum_{ij\in E_f\cap \delta(S)}u_{ij}^\mu-\sum_{i\in S}b_i^\mu\notag\\
&=\sum_{ij\in F_2}u_{ij}\left(\frac{\gamma_{ij}}{\mu_j}-\frac{1}{\mu_i}\right)-\sum_{ij\in E_f\cap \delta(S)}\frac{u_{ij}}{\mu_i}-\sum_{i\in S}\frac{b_i}{\mu_i}\label{eq:optimal}
\end{align}
Let $B^*\le B^m$ denote the common denominator of all $\gamma_{ij}$'s. We claim that every term in the above expression is an integer multiple of $1/B^*$.
Indeed, using that $\mu$ is a canonical labeling for $f$, there exists a
tight path from each node $i$ to a negative node $t$.
$1/\mu_i$ is then the product of the integer $M_t$ and the gain factors on such a tight path and is hence an integer multiple of $1/B^*$.
Similarly, every $\gamma_{ij}/\mu_j$ is also an integer multiple of $1/B^*$. Since the $u_{ij}$'s and $b_i$'s are integers, this verifies the claim.
Consequently, $0\le Ex^\mu(f')\le Ex^\mu(f)\le 1/B^m\le 1/B^*$. This implies $Ex^\mu(f')=0$, completing the proof.
\end{proof}

}{%
\section{Previous algorithms for flow problems}\label{sec:background}
We refer the reader to \cite{amo} for a background on flow problems; the full version also gives a more detailed overview. The fundamental problem for our investigations is the {\bf minimum-cost circulation} problem.\footnote{We shall use the term `circulation' to distinguish form other flow problems in the paper.} Algorithms are built on two main algorithmic paradigms: {\sl cycle cancelling} (see e.g. \cite[Chapter 9.6]{amo}) 
and  \textsl{successive shortest paths} (see e.g. \cite[Chapter 9.7]{amo}).
Neither of these basic algorithms are polynomial, but both can be modified to run in strongly polynomial time (e.g. \cite{Tardos85strong,Goldberg89,Orlin93}). For the successive shortest path framework, the first (weakly) polynomial running time was obtained by the scaling method of Edmonds and Karp \cite{Edmonds72}. This serves as a starting point for a significant part of algorithms for various flow models, including our concave generalized flow algorithm.

\medskip
\noindent
{\bf Generalized flow algorithms} are based on methods for minimum-cost circulations. A key notion is that of a {\sl flow generating cycle}, where the product of the gain factors $\gamma_e$ is greater than 1. This corresponds to negative cycles with respect to 
the cost  $c_e=-\log \gamma_e$. A solution is optimal, if there exists no flow generating cycle in the residual graph, connected to the sink by a path.
We may cancel all flow generating cycles by directly adapting algorithms for minimum-cost circulations.
Onaga \cite{Onaga67} showed that if after cancelling all flow generating cycles, we only use highest gain augmenting paths for excess transportation, no new flow generating cycle is created. This is analogous to the successive shortest paths algorithm and is also not polynomial.
The \textsc{Fat-Path} algorithm by  Goldberg, Plotkin and Tardos \cite{Goldberg91} uses a method analogous to the Edmonds-Karp capacity scaling. 
In the $\Delta$-phase, instead of using highest gain paths, $\Delta$-fat paths are used, that are able to transport $\Delta$-units of excess. This may create new flow generating cycles, which should be canceled in the next phase.

The basic framework of \cite{Onaga67} and of \textsc{Fat-Path}, namely using the different paradigms for eliminating flow-generating cycles and for transporting excess to the sink has been adopted by most subsequent algorithms,
e.g. \cite{Goldfarb96,Goldfarb97,Tardos98,Fleischer02,Radzik04}.

\medskip

\noindent
\textbf{Minimum-cost circulations with separable convex costs}.
A natural and well-studied nonlinear extension of minimum-cost circulations is replacing each arc cost $c_e$ by
a convex function $C_e$. This is a widely applicable  framework, see \cite[Chapter 14]{amo}.
Both the minimum mean cycle cancellation and the capacity scaling algorithms can be naturally extended to this problem with polynomial running time bounds: cycle cancellation was adapted by Karzanov and McCormick \cite{Karzanov97},
while capacity scaling by Minoux \cite{Minoux86} and by Hochbaum and Shanthikumar \cite{Hochbaum90}.
 The two frameworks are based on fundamentally different relaxations of the KKT-conditions. \cite{Karzanov97} directly uses the (right) derivative values of the $C_e$'s, while \cite{Minoux86} and \cite{Hochbaum90} use a gradually refined linear approximation.

\medskip

\noindent
{\bf Concave generalized flows.}
Shigeno's \cite{Shigeno06} approach was to extend the \textsc{Fat-Path} algorithm of \cite{Goldberg91} to the concave setting.
 However, \cite{Shigeno06}  obtains polynomial running time bounds only for restricted classes of gain functions. The algorithm consists of two procedures applied alternately, similarly to \textsc{Fat-Path}: a cycle cancellation phase to generate excess on cycles with positive gains, and a path augmentation phase to transport new excess to the sink in chunks of $\Delta$. For both phases, previous methods naturally extend:
cycle cancelling is performed analogously to \cite{Karzanov97}, whereas 
path augmentation to \cite{Hochbaum90}. Unfortunately, this yields polynomial running time only under certain
restrictions. The main reason for this is that the different relaxations cannot be fit smoothly into a unified framework.
}{}

\section{Concave generalized flows algorithm}\label{sec:concave}
\nottoggle{full}{
Our algorithm for the case of linear gains does excess transportation similarly to  \textsc{Fat-Path}, however, the cycle-cancelling steps are completely eliminated and we use a purely scaling framework. 
The successive shortest paths algorithms for minimum cost circulations start with an infeasible pseudoflow, having both positive and negative nodes. To use an analogous method for generalized flows, we have to give up the standard framework of algorithms where $e_i\ge 0$ is always maintained for all $i\in V-t$. This is the reason why we use the more flexible symmetric model: we start with possibly several nodes having $e_i<0$, and our aim is to eliminate them. An important property of the algorithm is that we always have to maintain $\mu_i=1/M_i$ for $e_i<0$; for this reason we shall avoid creating new negative nodes.

Similarly to \textsc{Fat-Path}, we use a scaling algorithm.
In the $\Delta$-phase, we consider the residual graph restricted to $\Delta$-fat arcs, arcs that may participate in a highest gain $\Delta$-fat-path, and maintain a conservative labeling $\mu$ with $\gamma_{ij}^\mu\le 1$ on the $\Delta$-fat arcs. When moving to the $\Delta/2$-phase, this condition may get violated due to $\Delta/2$-fat arcs that were not $\Delta$-fat. Analogously to the Edmonds-Karp algorithm, we  modify the flow by saturating each violated arc and thereby restitute dual feasibility. However, these changes may create new negative nodes and thus violate the condition $\mu_i=1/M_i$ for $e_i<0$ we must  maintain.

We resolve this difficulty by maintaing a `security reserve' of $d_i\Delta\mu_i$ in each node $i$ ($d_i$ is the number of incident arcs). This gives an upper bound on the total change caused by restoring feasibility of incident arcs in all subsequent phases. 
We call a node $\Delta$-positive  if $e_i>d_i\Delta\mu_i$, $\Delta$-negative if $e_i<d_i\Delta\mu_i$ and $\Delta$-neutral if $e_i=d_i\Delta\mu_i$.
$\Delta$-negative nodes may become negative ($e_i<0$) at a later phase, and therefore we maintain the stronger condition
$\mu_i=1/M_i$ for them. We send flow from $\Delta$-positive nodes to $\Delta$-negative and $\Delta$-neutral ones. Thereby we treat some nodes
with $e_i>0$ as sinks and increase their excess further; however, as $\Delta$ decreases, such nodes may gradually become sources.

This linear generalized flow algorithm smoothly extends to concave generalized flows. We use the local linearization 
 $\theta_{\Delta}^\mu(ij)$ of $\Gamma_{ij}$ used by Shigeno, analogously to \cite{Hochbaum90}.
 In the $\Delta$-phase, we consider the graph of $\Delta$-fat arcs, and maintain $\theta_{\Delta}^\mu(ij)\le 1$ on them.

When moving from a $\Delta$-phase to a $\Delta/2$-phase in the linear algorithm, the only reason for infeasibility is due to  $\Delta/2$-fat arcs that were not $\Delta$-fat. In contrast, feasibility can be violated on $\Delta$-fat arcs as well, as  $\theta_{\Delta}^\mu(ij)\le 1<\theta_{\Delta/2}^\mu(ij)$ may happen
due to the finer linear approximation of the gain functions in the $\Delta/2$-phase. Fortunately, feasibility can be restored in this case as well, by changing the flow on each arc by a small amount.
}{%


\iftoggle{full}{
We describe the algorithm in the same structure as for the linear case: Section~\ref{sec:concave-opt} presents the optimality conditions; $\Delta$-conservative and $\Delta$-canonical labels are discussed in Sections~\ref{sec:concave-conserv} and \ref{sec:concave-canon}, respectively. 
Section~\ref{sec:concave-alg} presents the algorithm, and Section~\ref{sec:concave-anal} its analysis.
}{}
\subsection{Optimality conditions}\label{sec:concave-opt}
The characterization of optimality was given in \cite{Shigeno06}; we have to modify the results slightly as we use the symmetric formulation.
Let us call an arc $ij\in E$ {\em immense}, if $\Gamma_{ij}(\ell_{ij})=-\infty$, and other arcs {\em regular}.
First, let us transform the problem to an equivalent instance with {\em(i)} 
$\ell=0$ for every arc  and $\Gamma_{ij}(0)=0$ for every regular; and {\em(ii)} every gain function $\Gamma_{ij}$ is strictly monotone increasing  on $[0,u_{ij}]$. 

For {\em(i)},
on each arc $ij\in E$, let us replace $u_{ij}$ by $u_{ij}-\ell_{ij}$ and $\ell_{ij}$ by 0.
If $ij$ is a regular arc, we modify the gain function to $\Gamma_{ij}(\alpha+\ell_{ij})-\Gamma_{ij}(\ell_{ij})$, and if $ij$ is an immense arc, to
 $\Gamma_{ij}(\alpha+\ell_{ij})$. Accordingly for every $i\in V$, let us increase $b_i$ by $\sum_{j:ij\in E}\ell_{ij}$, and  
decrease it by the sum of  $\Gamma_{ji}(\ell_{ji})$'s on regular arcs.

For {\em(ii)}, let us define 
$\tilde{u}_{ij}=\inf\{p: 0\le p\le u_{ij}, \Gamma_{ij}(p)=\Gamma_{ij}(u_{ij})\}$. 
By concavity, $\Gamma(\tilde{u}_{ij})=\Gamma_{ij}(u_{ij})$, and $\Gamma_{ij}(u_{ij})$ is strictly monotone
increasing on the interval $[0,\tilde{u}_{ij}]$. Let us replace $u_{ij}$ by $\tilde {u}_{ij}\le u_{ij}$.

For a pseudoflow $f:E\rightarrow \R$, we define the residual network
 $G_f=(V,E_f)$ identical as for the generalized flow setting:
$ij\in E_f$ if $ij\in E$ and $f_{ij}<u_{ij}$ or $ji\in E$ and $f_{ji}>0$.
For notational convenience, we define $f_{ji}=-\Gamma_{ij}(f_{ij})$ on backward arcs.
We also define the function  $\Gamma_{ji}(\alpha): [-\Gamma_{ij}({u}_{ij}),-\Gamma_{ij}(0)]\rightarrow [-u_{ij},0]$ by
\[
\Gamma_{ji}(\alpha)=-\Gamma^{-1}_{ij}(-\alpha).
\]
Hence $\Gamma_{ji}(f_{ji})=-f_{ij}$.

\medskip

The concavity of $\Gamma_{ij}$  implies that for each $0\le \alpha<u_{ij}$, there exists the right
derivative, denoted by $\Gamma_{ij}^+(\alpha)$, and for $0< \alpha\le u_{ij}$, there exists the left
derivative $\Gamma_{ij}^-(\alpha)$. If $0<\Delta<\Delta'$, then 
\begin{subequations}
\begin{align}
\frac{\Gamma_{ij}(\alpha+\Delta')-\Gamma_{ij}(\alpha)}{\Delta'}\le 
\frac{\Gamma_{ij}(\alpha+\Delta)-\Gamma_{ij}(\alpha)}{\Delta}\le \Gamma_{ij}^+(\alpha),\label{eq:delta-plusz}\\
\frac{\Gamma_{ij}(\alpha)-\Gamma_{ij}(\alpha-\Delta')}{\Delta'}\ge
 \frac{\Gamma_{ij}(\alpha)-\Gamma_{ij}(\alpha-\Delta)}{\Delta}\ge \Gamma_{ij}^-(\alpha)\label{eq:delta-minusz}
\end{align}\label{eq:delta-pluszminusz}
\end{subequations}
\noindent for $0\le \alpha\le u_{ij}$ and for $\Delta'\le \alpha\le u_{ij}$, respectively.
Furthermore, if $0<\alpha<\alpha'< u_{ij}$, then
$\Gamma_{ij}^+(\alpha')\le \Gamma_{ij}^-(\alpha')\le \Gamma_{ij}^+(\alpha)\le \Gamma_{ij}^-(\alpha)$.
The following claim is  easy verify.
\begin{claim}
For any $ij\in E$ with $0<f_{ij}< u_{ij}$, 
$\Gamma_{ij}^+(f_{ij})=1/\Gamma_{ji}^-(f_{ji})$, $\Gamma_{ij}^-(f_{ij})=1/\Gamma_{ji}^+(f_{ji})$. \trivi
\end{claim}

Let $P=i_0\ldots i_k$ be a walk in the auxiliary graph $E_f$. By sending $\alpha$ units of flow along $P$, we mean the following. First we increase $f_{i_0i_1}$ by $\alpha$ and set 
$\beta=\Gamma_{i_0i_1}(f_{i_0i_1}+\alpha)-\Gamma_{i_0i_1}(f_{i_0i_1})$ to be the flow arriving at $i_1$.
In step $h=1,\ldots,k-1$, we increase the flow on $i_hi_{h+1}$ by $\beta$ and set the new value of $\beta$
as $\Gamma_{i_hi_{h+1}}(f_{i_hi_{h+1}}+\beta)-\Gamma_{i_hi_{h+1}}(f_{i_hi_{h+1}})$. We assume $\alpha$ is chosen small enough so that no capacity gets violated. Let $f^{\alpha,P}$ denote the modified flow.

If $C=i_0\ldots i_{k-1}$ is a cycle if $E_f$, then by sending $\alpha$ units of flow around $C$ from $i_0$ 
we mean sending $\alpha$ units on the walk $i_0\ldots i_{k-1} i_0$. This modifies $e_i$ only in node $i_0$:
if the flow increase from $i_{k-1}i_0$ is bigger than $\alpha$, then $e_{i_0}$ increases, and if it is smaller
then it decreases. The next lemma characterizes when $e_{i_0}$ can increase.
For an arbitrary walk $P$ in $E_f$, let $\Gamma^+_f(P)=\Pi_{e\in P} \Gamma^+_{e}(f_e)$.
\begin{lemma}\label{lem:flow-gen}
Let $C$ be a cycle in $E_f$ with $i\in V(C)$. If $\Gamma^+_f(C)>1$ then $e_{i}$ can be increased by sending some flow around $C$. If $\Gamma^+_f(C)\le 1$, then it is not possible to increase $e_{i}$ by sending  any amount of flow around $C$. \nottoggle{full}{\trivi}{}
\end{lemma}
Since this property is independent from the choice of $i$, we simply say that $C$ is a \textsl{flow generating cycle} if $\Gamma^+_f(C)>1$. The lemma is an immediate consequence of the following claim.
\begin{claim}\label{cl:flow-path-inc}
Let $P=i_0i_1\ldots i_k$ be a walk in $E_f$. 
For any value of $\alpha>0$, the flow increase in $i_k$
for $f^{\alpha,P}$ is at most  $\Gamma^+_f(P)\alpha$.
On the other hand, for  any $\varepsilon>0$ there exists a $\delta>0$ so that for any $0< \alpha\le \delta$, $f^{\alpha,P}$ increases $e_{i_k}$ by at least $(\Gamma^+_f(P)-\varepsilon)\alpha$.
\end{claim}
\begin{proof}
The first part is trivial by concavity. 
We prove the second part by induction on the subpaths $P_h=i_0\ldots i_h$ for $h=0,\ldots, k$. 
There is nothing to prove for $h=0$; assume we have already proved it for $P_{h-1}$.
We want to find a $\delta>0$ for some an $\varepsilon>0$ satisfying the claim for $P_h$.
First, it is possible to pick a
small enough $\varepsilon^*>0$ such that
\begin{equation}
\Gamma^+_f(P_h)-\varepsilon<(\Gamma^+_{i_{h-1}i_h}(f_{i_{h-1}i_h})-\varepsilon^*) (\Gamma^+_f(P_{h-1})-\varepsilon^*).\label{eq:beta2}
\end{equation}
By the definition of $\Gamma^+_{i_{h-1}i_h}$,  there
exists a $\delta^*>0$ such that for any $0<\beta\le \delta^*$, 
\begin{equation}
(\Gamma^+_{i_{h-1}i_h}(f_{i_{h-1}i_h})-\varepsilon^*)\beta\le  \Gamma_{i_{h-1}i_h}(f_{i_{h-1}i_h}+\beta)-\Gamma_{i_{h-1}i_h}(f_{i_{h-1}i_h}).\label{eq:beta}
\end{equation}
By induction for $P_{h-1}$ and $\varepsilon^*$, we can choose a small enough $\delta>0$ with the following properties:
If $0<\alpha<\delta$, then $\Gamma^+_f(P_{h-1})\alpha\le \delta^*$, and
the increase of $e_{h-1}$ for $f^{\alpha,P_{h-1}}$ is at least
$\beta=(\Gamma^+_f(P_{h-1})-\varepsilon^*)\alpha$. Then (\ref{eq:beta2}) and (\ref{eq:beta}) show that 
 for $0< \alpha\le \delta$, $f^{\alpha,P}$ increases $e_{i_h}$ by at least $(\Gamma^+_f(P)-\varepsilon)\alpha$.
\end{proof}

The definition of GAPs is  analogous to the linear case, with the only difference that 
$\Gamma^+_f(P)>M_s/M_t$ instead of $\gamma(P)>M_s/M_t$ in case {\em(c)}.
The following lemma can be proved similarly as Lemma~\ref{lem:opt-no-GAP}.
\begin{lemma}\label{lem:conv-opt-no-GAP}
If $f$ is an optimal solution, then no GAP may exist. \trivi
\end{lemma}

\medskip
\textsl{Relabelings} are also defined analogously as for generalized flows.
Given $\mu: V\rightarrow \R_{>0}\cup\{\infty\}$, let us define $f_{ij}^\mu=f_{ij}/\mu_i$ for each arc $ij\in E$. We get problems equivalent to the original with relabeled functions
$\Gamma_{ij}^\mu(\alpha)=\Gamma_{ij}(\mu_i\alpha)/\mu_j$.
Accordingly, the relabeled demands, excesses, and capacities are
$b^\mu_i=b_i/\mu_i$, $e^\mu_i=e_i/\mu_i$, and
$u^\mu_{ij}=u_{ij}/\mu_i$.
A relabeling is \textsl{conservative}, if for any residual arc $ij\in E_f$,
 $\Gamma^{\mu+}_{ij}(f^\mu_{ij})\le 1$, that is, no edge may increase the relabeled flow.
Furthermore we require $\mu_i\ge 1/M_i$ for every $i\in V$ and equality whenever 
$e_i<0$.

We use the same convention for infinite $\mu_i$ values as for generalized flows.
If  $\mu_i=\infty$, we define $b^\mu_i=e^\mu_i=0$, $u_{ij}^\mu=0$ for $ij\in E$, and
furthermore  $\Gamma^{\mu+}_{ji}(f_{ji}^\mu)=0$ for all arcs $ji\in E_f$. Finally, for  $ij\in E_f$ with $\mu_i=\infty$, $\mu_j<\infty$, let $\Gamma^{\mu+}_{ij}(f_{ij}^\mu)=\infty$.

If $\mu$ is conservative, then if for a node $i\in V$  there exists a path from $i$ to a node $t\in V$
with $e_t<0$, then $\mu_i<\infty$. The following claim is also easy to verify.

\begin{claim}\label{cl:deriv-relabel}
$\Gamma^{\mu+}_{ij}(\alpha)=\frac{\mu_i}{\mu_j}\Gamma^+_{ij}(\alpha)$, and
$\Gamma^{\mu-}_{ij}(\alpha)=\frac{\mu_i}{\mu_j}\Gamma^-_{ij}(\alpha)$.\trivi
\end{claim}
This claim implies that for an $s-t$ walk $P$, $\Gamma^{\mu+}_f(P)=\frac{\mu_s}{\mu_t}\Gamma^+_f(P)$, and thus 
for a cycle $C$, $\Gamma^{\mu+}_f(C)=\Gamma^{+}_f(C)$.


\begin{theorem}[\cite{Shigeno06}]\label{thm:conc-opt}
For a pseudoflow $f$, the
following are equivalent.
\begin{enumerate}[(i)]
\item $f$ is an optimal solution to the symmetric version.
\item $E_f$ contains no generalized augmenting paths.
\item There exists a conservative labeling $\mu$ with  $e_i=0$ whenever $1/M_i<\mu_i<\infty$.
\end{enumerate}
%
\end{theorem}
\begin{proof}
The equivalence of {\em(i)} and {\em(iii)} follows by the Karush-Kuhn-Tucker conditions, with $\mu_i$ being the reciprocal of the Lagrange multiplier corresponding to $e_i+\kappa_i\ge 0$. {\em(i)} implies {\em(ii)} by Lemma~\ref{lem:conv-opt-no-GAP}.
The proof of {\em (ii)$\Rightarrow$(iii) } is the same as in Theorem~\ref{thm:genflow-opt},
with $\gamma(P)$ replaced by $\Gamma^+_f(P)$.
\end{proof}
}

\nottoggle{full}{
\subsection{Optimality conditions}\label{sec:concave-opt}
The characterization of optimality was given in \cite{Shigeno06}; we have to modify them slightly as we use the symmetric formulation.
The problem can be easily transformed to an equivalent instance with {\em(i)} 
$\ell=0$ for every arc  and $\Gamma_{ij}(0)=0$ for every arc with $\Gamma_{ij}(0)>-\infty$; and {\em(ii)} every gain function $\Gamma_{ij}$ is strictly monotone increasing on $[0,u_{ij}]$. We shall assume these properties in the sequel.

The concavity of $\Gamma_{ij}$  implies that for each $0\le \alpha$, there exists the right
derivative, denoted by $\Gamma_{ij}^+(\alpha)$, and for $0< \alpha$, there exists the left
derivative $\Gamma_{ij}^-(\alpha)$. If $0\le \alpha<\alpha'$, then
$\Gamma_{ij}^+(\alpha')\le \Gamma_{ij}^-(\alpha')\le \Gamma_{ij}^+(\alpha)\le \Gamma_{ij}^-(\alpha)$.

For a pseudoflow $f:E\rightarrow \R$, we define the residual network
by $ij\in E_f$ if $ij\in E$ or $ji\in E$ and $f_{ji}>0$.
For notational convenience, we define $f_{ji}=-\Gamma_{ij}(f_{ij})$ on backward arcs.
We also define the function  $\Gamma_{ji}(\alpha): [-\Gamma_{ij}({u}_{ij}),\Gamma_{ij}(0)]\rightarrow [- u_{ij},0]$ by
$\Gamma_{ji}(\alpha)=-\Gamma^{-1}_{ij}(-\alpha)$.
Hence $\Gamma_{ji}(f_{ji})=-f_{ij}$.

In the concave setting, we call a cycle $C$ in $E_f$ a {\sl flow generating cycle}, if $\Gamma^+(C)=\Pi_{ij\in C}\Gamma^+_{ij}(f_{ij})>1$.
For such a $C$, it can be shown that positive flow can be generated in any node of $C$ by sending flow around the cycle. 
The pair $(C,P)$ is called a \textsl{generalized augmenting path (GAP)} in the following cases:
\textbf{(a)}
$C$ is a flow generating cycle, $i\in V(C)$, $t\in V$ is a node with $e_t<0$, and
 $P$ is a path in $E_f$ from $i$ to $t$ ($i=t$, $P=\emptyset$ is possible);
\textbf{(b)}
$C=\emptyset$, and
$P$ is a path in $E_f$ between two nodes $s$ and $t$ with $e_s>0$, $e_t<0$;
\textbf{(c)}
$C=\emptyset$, and
$P$ is a path in $E_f$ between  $s$ and $t$ with $e_s\le 0$, $e_t<0$ 
and  $\Gamma^+_f(P)>M_s/M_t$.

\begin{lemma}\label{lem:conv-opt-no-GAP}
If $f$ is an optimal solution, then no GAP exists. \trivi
\end{lemma}

\textsl{Relabelings} are also defined analogously as for generalized flows.
Given $\mu: V\rightarrow \R_{>0}\cup\{\infty\}$, let us define $f_{ij}^\mu=f_{ij}/\mu_i$ for each arc $ij\in E$. We get problems equivalent to the original with relabeled functions
$\Gamma_{ij}^\mu(\alpha)=\Gamma_{ij}(\mu_i\alpha)/\mu_j$.
Accordingly, the relabeled demands, excesses, and capacities are
$b^\mu_i=b_i/\mu_i$, $e^\mu_i=e_i/\mu_i$, and
$u^\mu_{ij}=u_{ij}/\mu_i$.
A relabeling is \textsl{conservative}, if for any residual arc $ij\in E_f$,
 $\Gamma^{\mu+}_{ij}(f^\mu_{ij})\le 1$, that is, no edge may increase the relabeled flow.
Furthermore we require $\mu_i\ge 1/M_i$ for every $i\in V$ and equality whenever 
$e_i<0$.

If  $\mu_i=\infty$, we define $b^\mu_i=e^\mu_i=0$, $u_{ij}^\mu=0$ for $ij\in E$, and
furthermore  $\Gamma^{\mu+}_{ji}(f_{ji}^\mu)=0$ for all arcs $ji\in E_f$. Finally, if $ij\in E_f$ with $\mu_i=\infty$, $\mu_j<\infty$, then $\Gamma^{\mu+}_{ij}(f_{ij}^\mu)=\infty$.
The following theorem can be derived using the Karush-Kuhn-Tucker conditions.
\begin{theorem}[\cite{Shigeno06}]\label{thm:conc-opt}
Let $f\in \R^E$ satisfy $0\le f\le \tilde u$.
Then the following are equivalent.
\textbf{(i)} $f$ is an optimal solution to the symmetric version.
\textbf{(ii)} $E_f$ contains no generalized augmenting paths.
\textbf{(iii)} There exists a conservative labeling $\mu$ with  $e_i=0$ whenever $1/M_i<\mu_i<\infty$. \trivi
\end{theorem}
}{}

\iftoggle{full}{%
\subsection{$\Delta$-conservative labelings}\label{sec:concave-conserv}}{%
\subsection{$\Delta$-conservative and $\Delta$-canonical labelings}
}
\iftoggle{full}{We define the notion of $\Delta$-conservative labeling analogously as in Section~\ref{sec:delta-canon} for the linear case.}
Let us define the {\em \st} of $ij\in E_f$ by $s_f(ij)=\Gamma_{ij}(u_{ij})-\Gamma_{ij}(f_{ij})$ (if $ij$ is a backward arc, this is equivalent to 
$s_f(ij)=f_{ji}$.)
The \st expresses the maximum possible flow increase in $j$ if we saturate $ij$.
This notion enables us to identify arcs that can participite in fat paths during the algorithm.
In accordance with the other variables, the relabeled \st is defined as $s_f^\mu(ij)=s_f(ij)/\mu_j$.

Consider a scaling parameter $\Delta>0$.
\nottoggle{full}{As in the linear case, the {\em $\Delta$-fat graph} $E^\mu_f(\Delta)$ is the set of residual  arcs of relabeled \st at least $\Delta$.}{%
The {\em $\Delta$-fat graph} $E^\mu_f(\Delta)$ is the set of residual  arcs of relabeled \st at least $\Delta$: 
\[
E^\mu_f(\Delta)=\{ij: ij\in E_f: s_f^\mu(ij)\ge \Delta\}.
\]
Arcs in $E^\mu_f(\Delta)$ will be called {\em $\Delta$-fat arcs}.}
As in  \cite{Shigeno06}, we  use the following  linearization on $\Delta$-fat arcs in chunks of $\Delta$.
\begin{equation}
\theta^{\mu}_{\Delta}(ij):=
\frac {\Delta \mu_i}{\Gamma_{ij}^{-1}(\Gamma_{ij}(f_{ij})+\Delta \mu_j)-f_{ij}}\quad ij\in E_f^\mu(\Delta).\label{def:theta}
\end{equation}
\iftoggle{full}{This is well-defined since $\Gamma_{ij}(f_{ij})+\Delta\mu_{j}\le\Gamma_{ij}(u_{ij})$ for $\Delta$-fat arcs.
(\ref{def:theta}) can be written equivalently as
\begin{equation}
\theta^{\mu}_{\Delta}(ij)=
\frac\Delta{{\Gamma^\mu_{ij}}^{-1}(\Gamma^\mu_{ij}(f^\mu_{ij})+\Delta)-f_{ij}^\mu} \quad ij\in E_f^\mu(\Delta).\label{def:theta-2}
\end{equation}
Also, if $ji$ is a $\Delta$-fat arc, than using $\Gamma_{ji}(f_{ji})=-f_{ij}$ and $\Gamma_{ji}^{-1}=-\Gamma_{ij}(-\alpha)$, 
we get
\begin{equation}
\theta_{\Delta}^\mu(ji)=\frac {\Delta\mu_j}{\Gamma_{ij}(f_{ij})-\Gamma_{ij}\left(f_{ij}-\Delta\mu_i\right)}.\label{def:theta-reverse}
\end{equation}}{}
Consider a label function $\mu:V\rightarrow \R_{>0}\cup\{\infty\}$%
\nottoggle{full}{; recall that $d_i$ is the total number of arcs incident to $i$.}{.}
A node $i\in V$ is called \textsl{$\Delta$-negative} if $e^\mu_i<d_i\Delta$, \textsl{$\Delta$-neutral} if $e^\mu_i=d_i\Delta$ and
\textsl{$\Delta$-positive} if $e^\mu_i>d_i\Delta$.  
The labeling $\mu$ 
 {\em $\Delta$-conservative}, if $\theta_{\Delta}^\mu(ij)\le 1$ holds for every $ij\in E^\mu_f(\Delta)$. Furthermore, we require
$\mu_i\ge 1/M_i$ for all $i\in V$, with equality for every $\Delta$-negative node $i$.%

\nottoggle{full}{
Note that a $\Delta$-conservative labeling cannot  have any nodes with $\mu_i=\infty$.}{}
Using the convexity of $\Gamma^{-1}$, it can be shown that if $\mu$ is  a $\Delta$-conservative labeling then it is $\Delta'$-conservative
for all $\Delta'\ge \Delta$. 
\iftoggle{full}{Let $Ex^\mu(f)=\sum_{i\in V} \max\{e_i^\mu,0\}$  and $Ex_\Delta^\mu(f)=\sum_{i\in V} \max\{e_i^\mu-d_i\Delta,0\}$ denote the total relabeled
 excess of positive and $\Delta$-positive nodes, respecitively.}{%
Let
$Ex_\Delta^\mu(f)=\sum_{i\in V} \max\{e_i^\mu-d_i\Delta,0\}$ denote the total relabeled excess of $\Delta$-positive nodes.
}

The key importance of  $\Delta$-conservativity  is that it is maintained when sending $\Delta$ units of flow
on arcs with $\theta^\mu_\Delta(ij)=1$. This is  formulated in the next simple lemma.
\begin{lemma}\label{lem:conservative-preserved}
Assume $\mu$ is $\Delta$-conservative, and let $ij\in E_f^\mu(\Delta)$ be an arc with $\theta^\mu_\Delta(ij)=1$.  
If we increase $f^\mu_{ij}$ by $\Delta$, then $\Gamma(f^\mu_{ij})$ also increases by $\Delta$, and 
$\Delta$-conservativity is maintained.\nottoggle{full}{\trivi}{}
\end{lemma}
\iftoggle{full}{%
\begin{proof}
$\theta_\Delta^\mu(ij)=1$ is equivalent to
$\Gamma_{ij}^\mu(f^\mu_{ij})+\Delta)=\Gamma_{j}^\mu(f^\mu_{ij})+\Delta$, showing the first part.
Let $\bar f_{ij}=f_{ij}+\Delta\mu_i$ be the modified flow.
For the second part, $\theta_\Delta^\mu(ij)\le 1$ for $\bar f_{ij}$ easily follows  from convexity. 
Further, observe (\ref{def:theta-reverse}) shows that we get $\theta_\Delta^\mu(ji)=1$ for $\bar f_{ij}$.
This gives $\Delta$-conservativity for the modified flow as all other arcs are left unchanged. 
\end{proof}
}{}

\nottoggle{full}{%
The next lemma shows how a $\Delta$-conservative labeling can be transformed to a $\Delta/2$-conservative one.}{
In contrast to Lemma~\ref{lem:modify-Delta}, the following claim only enables to transform a $\Delta$-conservative labeling can be transformed to a $\Delta/2$-conservative one. The reason is that besides the set of $\Delta/2$-fat arcs being larger than the $\Delta$-fat arcs, we may have $\Delta$-fat arcs with $\theta_{\Delta}^\mu(ij)\le 1<\theta_{\Delta/2}^\mu(ij)$.}{}

\begin{lemma}\label{lem:modify-Delta-nonlin}
Let $f$ be a pseudoflow with a $\Delta$-conservative labeling $\mu$.
Then there exists a 
flow $\bar f$
such that $\mu$ is $\Delta/2$-conservative for $\bar f$ 
and $Ex_{\Delta/2}^\mu(\bar f)\le Ex_{\Delta}^\mu(f)+\frac 32m\Delta$.
\end{lemma}
\iftoggle{full}{
\begin{proof}
Consider a $\Delta/2$-fat arc $ij$ with $\theta_{\Delta/2}^\mu(ij)>1$ for $f$, that is,
\begin{equation}
 \Gamma^{-1}_{ij}\left(\Gamma_{ij}(f_{ij})+\frac\Delta2\mu_j\right)-f_{ij} <\frac{\Delta}2\mu_i.\label{eq:violate}
\end{equation}
There are two possible scenarios: {\em (a)} $ij$ was not $\Delta$-fat, that is,
\begin{equation}
\frac \Delta2\mu_j\le \Gamma_{ij}(u_{ij})-\Gamma_{ij}(f_{ij}) \le \Delta\mu_j,\label{eq:ujel}
\end{equation}
or {\em (b)} $ij$ was also a $\Delta$-fat arc. Then by $\Delta$-conservativity,
\begin{equation}
 \Gamma^{-1}_{ij}(\Gamma_{ij}(f_{ij})+\Delta\mu_j)-f_{ij} \ge \Delta\mu_i.\label{eq:regiel}
\end{equation}
In both cases, let us define
\[
\bar f_{ij}=\Gamma^{-1}_{ij}\left(\Gamma_{ij}(f_{ij})+\frac \Delta2 \mu_j\right).
\]
$\Delta/2$-fatness of $ij$ guarantees that this is well-defined.
%
In case {\em (a)}, we claim that $ij$ is not $\Delta/2$-fat for $\bar f$.
Indeed, 
\[
\Gamma_{ij}(u_{ij})-\Gamma_{ij}(\bar f_{ij})=\Gamma_{ij}(u_{ij})-\left(\Gamma_{ij}(f_{ij})+\frac \Delta2 \mu_j\right)< \frac \Delta2 \mu_j.
\]
The last inequality follows by the second part of (\ref{eq:ujel}).
In case {\em(b)}, we claim that if 
$ij$ is a $\Delta/2$-fat arc for $\bar f$ then $\theta_{\Delta/2}^\mu(ij)\le 1$ must hold for $\bar f$.
Indeed, if we subtract (\ref{eq:violate}) from (\ref{eq:regiel}), we get
\[
 \Gamma^{-1}_{ij}\left(\Gamma_{ij}(f_{ij})+\Delta\mu_j\right)-\Gamma^{-1}_{ij}\left(\Gamma_{ij}(f_{ij})+\frac\Delta2\mu_j\right)> \frac{\Delta}2\mu_i,
\]
and by substituting $\bar f_{ij}$, it follows that
\[
\Gamma^{-1}_{ij}\left(\Gamma_{ij}(\bar f_{ij})+\frac\Delta2\mu_j\right)-\bar f_{ij}> \frac{\Delta}2\mu_i,
\]
that is, $\theta_{\Delta/2}^\mu(ij)<1$ for $\bar f$.

 We next show that  if $ji$ is also a $\Delta/2$-fat arc for $\bar f$, then
$\theta_{\Delta/2}^\mu(ji)\le 1$ holds for $\bar f$. Indeed, using (\ref{def:theta-reverse}),
$\theta_{\Delta/2}^\mu(ji)\le 1$  for $\bar f$ is equivalent to 
\[
{\Gamma_{ij}(\bar f_{ij})-\Gamma_{ij}\left(\bar f_{ij}-\frac\Delta2\mu_i\right)}\ge \frac\Delta2\mu_j.
\]
Equivalently,
\[
{\Gamma_{ij}(f_{ij})+\frac \Delta2\mu_j-\Gamma_{ij}\left(\bar f_{ij}-\frac\Delta2\mu_i\right)}\ge \frac\Delta2\mu_j.
\]
Subtracting $\frac \Delta2\mu_j$, rearranging and applying the monotone increasing function $\Gamma_{ij}^{-1}$, we get
$f_{ij}\ge \bar f_{ij}-\frac\Delta2\mu_i$, that follows from (\ref{eq:violate}) by substituting $\bar f_{ij}$.

We define $\bar f_{ij}$ the above way whenever $ij$ is a $\Delta/2$-fat arc with $\theta_{\mu}(ij)>1$. (As a simple consequence of concavity, this cannot be the case for both $ij$ and $ji$.)
If this does not hold for neither $ij$ nor $ji$, then let $\bar f_{ij}=f_{ij}$. The next claim compares $f_{ij}$ and $\Gamma(f_{ij})$ to $\bar f_{ij}$ and $\Gamma(f_{ij})$.
\begin{claim}
$|\bar f_{ij}^\mu-f^\mu_{ij}|\le \frac \Delta2$ and $|\Gamma^\mu_{ij}(\bar f_{ij}^\mu)-\Gamma^\mu_{ij}(f^\mu_{ij})|\le \frac \Delta2$.
\end{claim}
\begin{proof}
There is nothing to prove if $\bar f_{ij}=f_{ij}$. Assume $f_{ij}$ was increased as above (the statement is equivalent for $ij$ and $ji$).
The first part is identical to  (\ref{eq:violate}).
By the definition of $\bar f_{ij}$,
\[
\Gamma_{ij}(\bar f_{ij})-\Gamma_{ij}(f_{ij})=\Gamma_{ij}(f_{ij})+\frac\Delta2\mu_j-\Gamma_{ij}(f_{ij})= \frac\Delta2\mu_j,
\]
giving the second part.
\end{proof}
For $\Delta/2$-conservativity, we also need to show that $\bar f$ has no $\Delta/2$-negative nodes with $\mu_i>1/M_i$. 
By the above claim, the total possible change of relabeled flow on arcs incident to $i$ is 
$d_i\Delta/2$. A node is nonnegative for  $\Delta$ if $e_i^\mu\ge d_i\Delta$ and for $\Delta/2$ if  $e_i^\mu\ge d_i\Delta/2$. Consequently, a nonnegative node cannot become $\Delta/2$-negative.

Finally, $Ex_{\Delta/2}^\mu(f)\le Ex_{\Delta}^\mu(f)+\sum_{i\in V}d_i\Delta/2$, and  each arc is responsible for creating at most $\Delta/2$ units of new excess.
This gives $Ex_{\Delta/2}^\mu(\bar f)\le Ex_{\Delta}^\mu(f)+\frac m2\Delta$, as required.
\end{proof}
}{We give the proof in the Appendix.  Analogous claims are proved in \cite{Minoux86} and \cite{Hochbaum90}.}
The subroutine {\sc Adjust}($\Delta$) performs the  simple modifications described in the proof
\iftoggle{full}{(in contrast to the linear case, this function has only one parameter).}{.}

\iftoggle{full}{

\subsection{$\Delta$-canonical labelings}\label{sec:concave-canon}}{

\medskip

}
Given a pseudoflow $f$ and a $\Delta$-conservative labeling $\mu$, the arc $ij\in E_f^\mu(\Delta)$ is called {\em tight} if $\theta^\mu_\Delta({ij})=1$. A directed path in $E_f^\mu(\Delta)$ is called tight if it consists of tight arcs. 
$\mu$ is a \textsl{$\Delta$-canonical} labeling, if from each node $i$ there exists a tight path
to a $\Delta$-negative or to a $\Delta$-neutral node. 
 {Such a path is approximately a highest gain $\Delta$-fat augmenting path.}
The subroutine \textsc{Tighten-Label}$(f,\mu,\Delta)$  returns a $\Delta$-canonical label $\mu'\ge \mu$ for a $\Delta$-conservative
label $\mu$. \iftoggle{full}{This is almost identical to the algorithm described in Section~\ref{sec:delta-canon}. The only difference is in the definition of the multiplier $\alpha$, which is given by (\ref{eq:alpha-def}) for the linear case.
Instead, we define
\begin{align*}
\alpha=\min\left\{
\min 
\left\{ \frac{1}{\theta_\Delta^\mu({ij})}:  ij\in E_f^\mu(\Delta)\right\},
\min\left\{\frac{e^\mu_i}{d_i\Delta}: i\in V-S\right\}
\right\}.
\end{align*}
In an iteration, we multiply every $\mu_i$ by $\alpha$ for $i\in V-S$, where $S$ is the set of nodes from which there exists a tight path to a $\Delta$-negative or a $\Delta$-neutral node. We claim that as in the linear case, this maintains $\Delta$-conservativity, and extends $S$ by at least one node.
This is a simple consequence of the fact that multiplying $\mu_i$ by $\alpha$ multiplies $\theta_\mu(ij)$ by $\alpha$ for every incident arc $ij$.

To verify that $\mu$ remains $\Delta$-conservative, we also have to check $\theta_\Delta^
\mu(ij)\le 1$ on all arcs $ij\in E_f^\Delta(\mu)$, $j\in V-S$. This follows by the convexity of $\Gamma_{ij}^{-1}$.}{%
 This is a multiplicative variant of Dijsktra's algorithm (see the full version for details).
In every iteration, let $S$ be the set of nodes from which there exists a tight path to a $\Delta$-negative or $\Delta$-neutral node.
We increase $\mu_i$ for every node in $V\setminus S$
at the same rate, until either a new arc becomes tight or a $\Delta$-positive node becomes $\Delta$-neutral.
In both cases, $S$ is extended.}

\subsection{The main algorithm}\label{sec:concave-alg}

\begin{figure}[htb]
\begin{center}
\fbox{\parbox{\textwidth}{
\begin{tabbing}
xxxxx \= xxx \= xxx \= xxx \= xxxxxxxxxxx \= \kill
\> \textbf{Algorithm} \textsc{Symmetric Concave Fat-Path}\\
\> \textbf{for} $i\in V$ \textbf{do} $\mu_i\leftarrow \frac{1}{M_i}$;\\
\>   \textbf{for} $ij\in E$ \textbf{do}  $f_{ij}\leftarrow u_{ij}$; \\
\>  $\Delta\leftarrow MU+1$;\\
\> \textbf{while }$(\ketenn)\Delta\ge {\varepsilon}$ \textbf{do}\\
\> \> \textbf{do}\\
\> \>  \> \textsc{Tighten-Label}$(f,\mu,\Delta)$;\\
\> \>  \> $D\leftarrow \{i\in V:  e_i>(d_i+1)\Delta\}$;\\
\> \> \> $N_0\leftarrow \{i\in V: e_i\le d_i\Delta\}$;\\
\> \> \> \textbf{pick} $s\in D$, $t\in N$ connected by a tight path $P$;\\
\> \> \> send $\Delta$ units of flow along $P$;\\
\> \>  \textbf{while} $D\neq \emptyset$;\\
\> \> \textsc{Adjust}$(\Delta)$;\\
\> \>  $\Delta\leftarrow \frac{\Delta}2$;\\
\> \textbf{return} $\varepsilon$-approximate optimal solution $f$. \\
\end{tabbing}
}}
\caption{The algorithm for symmetric concave generalized flows}\label{code:concave-genflow}
\end{center}
\end{figure}

Let us initialize $\mu_i=1/M_i$ for every $i\in V$, and $f_{ij}=u_{ij}$ for every $ij\in E$. \iftoggle{full}{(We set the upper bounds rather than the lower bounds 0 because $\Gamma_{ij}(0)=-\infty$ is allowed.)}{}
Let us pick the initial value $\Delta=MU+1$. 

The algorithm consists of $\Delta$-phases, and terminates with an $\varepsilon$-approximate solution
if $(\ketenn)\Delta<\varepsilon$. During the $\Delta$-phase, we maintain a pseudoflow $f$ and a  $\Delta$-conservative labeling $\mu$.
The $\mu_i$ values may only increase.
Let $D$ denote the set of nodes  $i$ with ${e}_i^\mu>(d_i+1)\Delta$. The $\Delta$-phase consists of iterations,  and terminates whenever $D$ becomes empty.
In each iteration, we update $\mu$ to a canonical labeling by calling \textsc{Tighten-Label}$(f,\mu,\Delta)$. If $D\neq\emptyset$ still holds, send $\Delta$ units of relabeled flow on a tight path from some $s\in D$ 
to a $\Delta$-neutral or $\Delta$-negative node $t$.

\subsection{Analysis}\label{sec:concave-anal}
\begin{claim}\label{cl:gen-initial}
The initial $\mu$ is $\Delta$-conservative, and $\Delta$-conservativity is maintained during the entire $\Delta$-phase. {\trivi}{}
\end{claim}
\iftoggle{full}{\begin{proof}
Initially,  $f\equiv u$ and hence $E_f$ is the set of backward arcs. 
For an arc $ij\in E$, $s_\Delta^\mu(ji)=u_{ij}/\mu_i\le MU<\Delta$, and hence $E_f^\mu(\Delta)=\emptyset$.
Also, $\mu_i=1/M_i$ holds for every node $i$.
\textsc{Tighten-Label}$(f,\mu,\Delta)$ clearly maintains
$\Delta$-conservativity. 
We use only tight arcs to send flow, and Lemma~\ref{lem:conservative-preserved} guarantees that this preserves $\Delta$-conservativity.
At the end of the $\Delta$-phase, \textsc{Adjust}$(\Delta)$ transforms $f$ to a $\Delta/2$-conservative pseudoflow.
\end{proof}}{}

\begin{claim}\label{cl:preproc}
The $\Delta$-phase starts with  $Ex^\mu_\Delta(f)\le (\ketenn)\Delta$.
\end{claim}
\begin{proof}
For the initial solution, 
 $Ex^\mu_\Delta(f)\le M(\sum_{i\in V}|b_i|+mU)\le (m+n)MU$. The claim follows, since  $\Delta=MU+1$.
Once we finish all iterations in the $\Delta$-phase, $D=\emptyset$ implies $Ex^\mu_\Delta(f)\le n\Delta$.
By Lemma~\ref{lem:modify-Delta-nonlin}, \textsc{Adjust}$(\Delta)$ transfroms $f$ to a $\Delta/2$-conservative solution
by increasing the excess by at most $\frac 32m\Delta$.
Hence the $\Delta/2$ phase starts with $Ex^\mu_\Delta(f)\le (\ketenn)\Delta/2$, proving the claim.
\end{proof}

\begin{lemma}\label{lem:concave-push-bound}
A $\Delta$-phase consists of at most $\konstlin$ iterations.
\end{lemma}
\begin{proof}
Consider the potential function $\Psi=\sum_{i\in V}\lfloor \max\{{e}^\mu_i-(d_i+1)\Delta,0\}/\Delta\rfloor$. 
By Claim~\ref{cl:preproc}, $\Psi\le\konstlin$ holds at the beginning.  
In the relabeling steps, $\Psi$ may only decrease, and in every path augmentation, it decreases by exactly 1.
\end{proof}

Recall that $\kappa_f=\sum_{i\in V}M_i\kappa_i=\sum_{i\in V} M_i\min\{-e_i,0\}$ denotes the excess discrepancy.
For a $\Delta$-conservative $\mu$,   $M_i\kappa_i=e_i^\mu=$ holds for every node $i$ with $e_i<0$, because of  $\mu_i=1/M_i$.
Consequently, $\kappa_f$ is the total relabeled deficiency of the negative nodes.
The next theorem shows that if $\Delta<\varepsilon/(\ketenn)$, then we have an $\varepsilon$-optimal solution at the end of the $\Delta$-phase.
\nottoggle{full}{The key part of the proof is showing that if the algorithm runs forever, $f$ converges to an optimal solution $f^*$.}{}
\begin{theorem}\label{thm:delta-opt}
At the end of phase $\Delta$, the actual $f$ is $(\ketenn)\Delta$-optimal. \nottoggle{full}{\trivi}{}
\end{theorem}
\iftoggle{full}{
\begin{proof}
Let us keep running the algorithm forever unless it finds a 0-discrepancy solution at some phase.
First, consider the case when for some $\Delta'=\Delta/2^k$, we terminate with a 0-discrepancy solution. 
In all phases between $\Delta$ and $\Delta'$ , the total decrease of excess discrepancy is bounded by $(\ketenn)(\Delta/2+\Delta/4+\ldots+\Delta/2^k)< (\ketenn)\Delta$. Since in the $\Delta'$-phase we have a 0-discrepancy solution,  the total discrepancy at the end of the  $\Delta$-phase is at most $(\ketenn)\Delta$, proving the theorem.

Assume now the procedure runs forever. 
For each $i\in V$, $\kappa_i$ is decreasing and thus converges to some limit $\kappa_i^*$. 
Let $\kappa^*=\sum_{i\in V}M_i\kappa_i^*$. As above, the total decrease of the excess discrepancy after phase $\Delta$ is bounded by $(\ketenn)\Delta$, hence
$\kappa_f\le \kappa^*+(\ketenn)\Delta$. 
The proof finishes by constructing an optimal
pseudoflow $f^*$ with discrepancy $\kappa^*$.

Let $f^{(t)}$ denote the flow at time $t$, for  $\Delta^{(t)}=\Delta_0/2^t$, with labels $\mu_i^{(t)}$.
For each node $i$, $\mu_i^{(t)}$  is increasing; let $\mu_i^*=\lim_{t\rightarrow\infty} \mu_i^{(t)}$. 
For every $ij\in E$, $f_{ij}^{(t)}$ is  a bounded sequence ($0\le f_{ij}^{(t)}\le  u_{ij}$).
Consequently, we can choose an infinite set $T'\subseteq \mathbb{N}$ so that restricted to $t\in T'$, all
sequences $f_{ij}^{(t)}$ converge; let $f_{ij}^*$  denote the limits.
We shall prove that $f^*$ is an optimal pseudoflow with optimal labeling $\mu_i^*$,
completing the proof.

Let $V_\infty=\{i: \mu_i^*=\infty\}$. 
We claim that $V-V_\infty\neq\emptyset$. Indeed, if $i$ is $\Delta$-negative in a certain phase, then $\mu_i=1/M_i$, and
once $i$ becomes $\Delta$-positive or neutral, it would never again become $\Delta$-negative. Consequently, the set of $\Delta$-negative nodes is decreasing. Once it becomes empty,  we arrive at a 0-discrepancy solution. If it never becomes empty, then we have a set $N^*$ which remains the set of $\Delta$-negative nodes after a finite number of steps and thus $\mu_i^*=1/M_i$ for $i\in N^*$.

Let $e_i^*$ denote the  excesses of $f^*$
If $e_i^*<0$, then clearly $i\in N^*$ and $\mu_i^*=1/M_i$. If  $e_i^*>0$, we shall prove $\mu_i^*=\infty$.
For a contradiction, assume $\mu_i^*<\infty$. Then for sufficiently large $t\in T'$, $(d_i+\ketenn)\Delta^{(t)}\mu_i^{(t)}<e_i^{(t)}$ and thus $Ex^\mu_{\Delta^{(t)}}(f)>(\ketenn)\Delta^{(t)}$, a contradiction.

We have to prove $\Gamma^{\mu^*+}_{ij}(f_{ij}^{*\mu^*})\le 1$ whenever $ij\in E_{f^*}$.
If $\mu_{j}^*<\infty$, then $\Gamma^{\mu^*+}_{ij}(f_{ij}^{*\mu^*})=0$. If $\mu_{j}^*<\infty$, then
the definition (\ref{def:theta}) gives
\begin{equation*}
1\ge \theta^{\mu^{(t)}}_{\Delta^{(t)}}(ij)=\frac{\Delta^{(t)}\mu_j^{(t)}}{\Gamma_{ij}^{-1}(\Gamma_{ij}(f_{ij}^{(t)})+\Delta^{(t)}\mu_j^{(t)})-f_{ij}^{(t)}}\cdot \frac{\mu_i^{(t)}}{\mu_j^{(t)}}.
\end{equation*}
Then $\Delta^{(t)}\mu_j^{(t)}\rightarrow 0$ and hence the  first fraction converges to $\Gamma_{ij}^+(f_{ij}^*) =1/\{\Gamma^{-1}_{ij}\}^+(\Gamma_{ij}^+(f_{ij}^*) )$, while the second to $\mu_i^*/\mu_j^*$, leading to the conclusion using Claim~\ref{cl:deriv-relabel}. 
\end{proof}
}{}
\begin{theorem}\label{thm:concave-running-time-bound}
The above algorithm finds an $\varepsilon$-approximate solution to the symmetric concave generalized flow problem in $O(m(m+n\log n)\log(MUm/\varepsilon))$ oracle calls.
\end{theorem}
\begin{proof}
The initial value of $\Delta$ is $MU+1$, and we terminate if $\Delta<\varepsilon/(\ketenn)$ by Theorem~\ref{thm:delta-opt}. Hence the total
number of scaling phases is $O(\log(MUm/\varepsilon))$. The number of iterations in a phase is $O(m)$ by Lemma~\ref{lem:concave-push-bound}, and the running time of an iteration is dominated by \textsc{Tighten-Label}, a slightly modified version of Dijkstra's algorithm that can be implemented in $O(m+n\log n)$ time using Fibonacci heaps as in \cite{Fredman87}
\end{proof}

\iftoggle{full}{%
\section{Sink versions of the problems}\label{sec:sink}
In this section, we show how the algorithms in Sections~\ref{sec:gen-alg} and \ref{sec:concave} can be applied to solve the
to the sink versions of the corresponding problems. 
For generalized flows, let us set $M_t=1$ and $M_i=B^n+1$ for every $i\in V-t$.
Let us set $b_t=\lceil\sum_{j:jt\in E}\gamma_{jt}u_{jt}-\sum_{j:tj\in E}\ell_{tj}+1\rceil\le d_tB^2+1$. This a strict upper bound on $\sum_{j:jt\in E}\gamma_{jt}f_{jt}-\sum_{j:tj\in E}f_{tj}$, hence $e_t< 0$ will hold for any pseudoflow.

Let us run the algorithm for the symmetric formulation with these $M_i$'s, returning an optimal pseudoflow $f$ and optimal labels $\mu$.
We claim that $f$ is also optimal for the sink formulation. If $e_i\ge 0$ for all $i\neq t$, this is clearly the case. 

On the other hand, we claim that if there exists a node $i$ with $e_i<0$, then the sink version is infeasible.
First, we show that such an $i$ cannot be reached from $t$ on a path in $E_f$.
 Indeed, if $P$ were a $t$--$i$ path in $E_f$, 
then $1\ge\gamma^\mu(P)=\gamma(P)\mu_t/\mu_i$. Since both $i$ and $t$ are negative, $\mu_i=1/M_i=1/(B^n+1)$ and $\mu_t=1$.
Consequently, $\gamma(P)\le 1/(B^n+1)$. This is a contradiction since $\gamma(P)$ is the product of at most $n$ rational numbers, each with denominator at most $B$.
Let $V'\subseteq V$ be the set of nodes $j$ for which there exists a path in $E_f$ from $j$ to some $i\in V-t$ with $e_i<0$. 
This set verifies that the sink version cannot be feasible.

By setting the $b_t$ value and the $M_i$'s, $B$ has increased to $d_t B^2+1$ and $M=B^n+1$.
This gives running time $O(m^2(m+n\log n)\log B)$.

\medskip

Let us turn to concave generalized flows.%
}{%
\section{Sink version of the problem}\label{sec:sink}
Let us now show how the algorithm for the symmetric version can be used to solve the sink version.}
An $\varepsilon$-approximate solution to the sink version means a pseudoflow $f$
with $\sum_{i\in V-t}\max\{0,-e_i\}\le \varepsilon$ and $e_t$ being at least the optimum value minus $\varepsilon$.

Let us set $b_t=U^*+1$, a strict upper bound on $\sum_{j:jt\in E}\Gamma_{jt}(f_{jt})-\sum_{j:tj\in E}f_{tj}$%
\iftoggle{full}{ (we defined $U^*$ in Section~\ref{sec:complexity}).}{.}
 Thus $e_t<0$ is always guaranteed.
Let us set $M_i=\lceil 2U^*/\varepsilon\rceil+1$ if $i\in V-t$ and $M_t=1$. Let us run the algorithm for the symmetric formulation to obtain an $\varepsilon$-optimal solution $f$.

If $\kappa_f>2U^*+\varepsilon$, then no feasible solution may exist. Indeed, by the definition of $U^*$, if there is a feasible solution $f'$, then there exists one with $e_t\ge -U^*$.
If $f'$ is such a feasible solution for the sink formulation, then its excess discrepancy for the symmetric formulation is at most $\kappa_{f'}\le b_i+U^*\le 2U^*$, a contradiction as $f$ was $\varepsilon$-optimal for the symmetric formulation.

If $\kappa_f\le 2U^*+\varepsilon$, then 
\[
\sum_{i\in V-t}\max\{0,-e_i\}=\frac 1{\lceil 2U^*/\varepsilon\rceil+1}\sum_{i\in V-t}M_i\kappa_i\le \frac{\kappa_f}{\lceil 2U^*/\varepsilon\rceil+1}\le \varepsilon.
\]
Also $\kappa_t$ cannot be further than $\varepsilon$ from the optimum value of $e_t$ for the sink formulation.
Indeed, let $f'$ be the optimal solution to the sink formulation with $e'_t$ flow reaching the sink. Then $\kappa_{f'}=b_t-e'_t$.
The claim follows by 
\[b_t-e_t'+\varepsilon=\kappa_{f'}+\varepsilon\ge \kappa_f\ge \kappa_t=b_t-e_t,
\] 
and thus $e_t\ge e'_t-\varepsilon$.
This gives a running time bound  $O(m(m+n\log n)\log(U^*m/\varepsilon))$.

\iftoggle{full}{
\section{Finding the optimal solution for rational convex programs}\label{sec:market}
In this section, we first give a general theorem which shows how an approximate solution to the sink version can be converted to an exact optimal solution, given that one exists.
We shall verify the required technical properties with appropriate parameters for nonsymmetric Arrow-Debreu Nash bargaining. 
Unlike the linear Fisher model, ADNB might be infeasible. However, it can be shown that if the problem is infeasible, then for appropriate (polynomially small) $\varepsilon$, the $\varepsilon$-approximate version is also infeasible.
Similar reductions should hold for all other rational convex programs 
discussed in Section~\ref{sec:market-app} as well, giving polynomial time algorithms for finding  optimal solutions.

\begin{theorem}\label{thm:conv}
Let problem $\cal P$ be given by the sink formulation with $n$ nodes and $m$ arcs, and complexity parameters $U$, $U^*$. Assume $\cal P$ is guaranteed to have a rational optimal solution,
and the following conditions hold for some values $\varepsilon$, $T$ and a function $\tau(n,m,U^*)$.
\begin{itemize}
\item[(P1)] Consider the algorithm for  the sink version for an $\varepsilon$-approximation.
Then either there is no $\varepsilon$-feasible solution, or 
 $\mu_i\le T$ holds for any $i\in V$, even if running the algorithm for
 an arbitrary number of phases.
\item[(P2)] 
 A subroutine is provided for finding  an optimal solution $\tilde f$ in $\tau(n,m,U^*)$ time, if the following assumptions hold.
Assume that for each $ij\in E$, we are given an interval $I_{ij}\subseteq [\ell_{ij},u_{ij}]$ with $|I_{ij}|\le 2T\varepsilon$, with the guarantee that there exists an optimal solution $f^*$ with $f^*_{ij}\in I_{ij}$ for all $ij\in E$. 
\end{itemize}
Then there exist an algorithm for finding the exact optimal solution or proving that the problem is infeasible in $O(m(m+n\log n)\log(U^*m/\varepsilon))+\tau(n,m,U^*)$.
\end{theorem}
We remark that in {\em(P2)}, $\tilde f= f^*$ is not required.
\begin{proof}
Let us formulate the symmetric version for $\varepsilon$-approximation as in Section~\ref{sec:sink}.
Assume  we run the algorithm for this problem forever, as in the proof of Theorem~\ref{thm:delta-opt}. 
The $\mu_j$'s shall converge to
some finite values $\mu_j^*\le T$. In any  $\Delta$-phase, the total change of $f^\mu_{ij}$  is bounded
by $\varepsilon'=(\ketenn)\Delta$, and thus $f_{ij}$ may change by at most $T\varepsilon'$. Therefore all $f_{ij}$'s converge to some values $f^*_{ij}$, which
can be seen to give an optimal solution, as in the proof  of Theorem~\ref{thm:delta-opt}. 

The algorithm terminates whenever $\Delta<\varepsilon/(\ketenn)$.
At this point, the intervals $I_{ij}=[f_{ij}-T\varepsilon,f_{ij}+T\varepsilon]$ satisfy the conditions in {\em(P1)}, since  $|f_{ij}-f^*_{ij}|\le T\varepsilon$.
Running the $\varepsilon$-approximation algorithm and then the algorithm in {\em(P2)} gives the running time bound.
\end{proof}

To ensure property {\em(P1)}, a useful method is to enforce the existence of a unique optimal solution by perturbing the input data, as done by Orlin \cite{Orlin10} for linear Fisher markets. If there is a unique rational optimal solution $f^*$ with all entries having denominator at most $Q$, then setting $2T\varepsilon<1/Q$ enables us to identify the set of arcs with $f^*_{ij}>0$. This can be already enough to compute $f^*$ efficiently.

\subsection{Application to nonsymmetric Arrow-Debreu Nash bargaining}
Let us now apply Theorem~\ref{thm:conv} for the nonsymmetric ADNB problem.
Let us assume all utilities $U_{ij}$, budgets $m_i$ and disagreement utilities $c_i$ are nonnegative integers, with 
$U_{\max}=\max\{U_{ij}: i\in B, j\in G\}$,
$R=\max\{m_i: {i\in B}\}$, and $C=\max\{c_i: {i\in B}\}$. 
Let $n=|G|+|B|$ and let $m$ be the number of pairs $ij$ with $U_{ij}>0$;  in the concave generalized flow instance, the number of nodes is $n+1$ and the number of arcs is $m+|B|$. Let us assume that there exists at least one arc with positive utility  incident to any buyer and to any good. The special case $c\equiv 0$ is identical to Fisher's market with linear utilities. 

Consider a candidate solution with price $p_j$ for each good $j\in G$. Let $x_{ij}\ge 0$ denote the amount of good $j$ purchased by buyer $i$. It follows from the KKT-conditions (see also \cite{Vazirani11}) that $(p,x)$ is an optimal solution if and only if
{\em(i)} $\sum_{i\in B}x_{ij}=1$ for each good $j$, that is, each good is fully sold; and
{\em(ii)} for any buyer $i$ and good $j$, $U_{ij}/p_j\le (\sum_{j\in G}U_{ij}x_{ij}-c_i)/m_i$, and equality holds whenever
$x_{ij}>0$.


By strict concavity of the objective, the utilities $\sum_{j\in G}U_{ij}x_{ij}$ accrued
by the players are the same in any optimal solution whenever there is a feasible solution. Yet these same values can be obtained by different allocations.
As in \cite{Orlin10}, we assume that there exist a unique optimal allocation as well.
This can be done by a  lexicographic perturbation of the $U_{ij}$ values,
without significantly increasing the running time.  This guarantees that the set
of arcs with $x_{ij}>0$ in an optimal solution is cycle-free.

For the concave generalized flow instance, let us set upper capacity $v_{ji}=2$ on each arc $ji$ with $j\in G$, $i\in B$, and $v_{it}=2\sum_{j\in G}U_{ij}$ on each arc $it$ (we set larger capacities so that the capacity constraints would never become tight). 
Hence the complexity parameter $U$ is bounded by $2|G|U_{\max}$. We shall prove the following.

\begin{theorem}\label{thm:reduction}
Let $K=nRU_{\max}$. Setting $T=U^*=\max\{C,nK\log K\}$, $\varepsilon=1/(2K^nU^*)$ satisfies the
requirements on $U^*$ in Section~\ref{sec:complexity} and (P1) and (P2) in Theorem~\ref{thm:conv}.
Our algorithm delivers an optimal solution in  running time 
 $O(m(m+n\log n)(n\log (nRU_{\max})+\log C))$ for nonsymmetric ADNB.
\end{theorem}
The running time $\tau(n,m,U^*)$ will be negligible compared to the main algorithm and therefore it does not affect the complexity bound.
In the rest of the section, we shall verify the above choices. 

\begin{lemma}[{see \cite[Thm 2]{Vazirani11}, \cite[Lemma 2.1]{Orlin10}}]\label{lem:fisher-unique}
Assuming that the problem is feasible and there exists a unique optimal allocation $x^*$, all positive $x^*_{ij}$ values 
are rational numbers with  a common denominator $S\le K^n$.
\end{lemma}
\begin{proof}
The optimal allocations $x^*$ and prices $p^*$ can be uniquely obtained given the set $F$ of 
arcs $ji$ with $x^*_{ij}>0$.
If we introduce the variable $q_j=1/p_j$, then an optimal solution 
must satisfy the following system of linear equations.
\begin{align}\label{system:eq}
\sum_{j:ji\in E} x_{ij}&=1 &\forall i\in G\\
\sum_{k:ki\in E} U_{ik}x_{ik}-U_{ij}m_iq_j&=c_i &\forall ji\in F\nonumber
\end{align}
We claim that this system has a unique solution $(x^*,p^*)$.
Fixing the price of one good $j_0$ as $p_{j_0}=\alpha$ in a component of $F$, it uniquely determines $p_j$ for any good $j$ 
in the same component as $\alpha$ times the product of the $U_{ab}$ values on the unique path from $j_0$ to $j$. 
Similarly, this determines the best bang-per-buck values $b_i=U_{ij}/p_j=(\sum_{k:ki\in E} U_{ik}x_{ik}-c_i)/m_i$ for $ji\in F$, which are proportional to $1/\alpha$.
The optimality conditions imply that in an equilibrium, the money spent by buyer $i$ is $r_i=m_i+c_i/b_i$. In each component of $F$, the sum of prices should be equal to the money spent by the buyers. This uniquely determines all prices  and bang-per-buck values in the component. 
The $x_{ij}$ values in the component have to sum up to 1 for each good $i$ and 
$\sum_{k:ki\in E} U_{ik}x_{ik}=b_im_i+c_i$. As $F$ is a forest, this has a unique solution.

In the solution to (\ref{system:eq}), a common denominator is the determinant $S$ of a largest non-singular submatrix of the constraint matrix. The Hadamard-bound gives $S\le (nU_{\max}R)^n=K^n$, using that $|F|\le n-1$.
\end{proof}
The above proof also gives a simple linear time algorithm for finding the optimal solution, verifying {\em(P2)} 
if $2T\varepsilon<1/K^n$, with $\tau(n,m,U^*)$ being negligible compared to the running time of the approximation algorithm.

Next we justify the choice of $U^*$. $U\le U^*$ clearly holds. For an arbitrary pseudoflow $f$,
$e_t\le \sum_{i\in G} m_i\log v_{it}\le nR\log U\le U^*$. It is left to show that if the problem is feasible,
there exists a feasible solution with $e_t>-U^*$.
Since the $U_{ij}$ and $c_{i}$ values are integers, whenever 
$\sum_{j:ji\in E}U_{ij}x_{ij}-c_i>0$, it should be at least $1/K^n$. Consequently, if the sink version of the problem is feasible, the optimal objective value is at least $e_t\ge \sum_{i\in G} m_i\log (1/K^n)\ge -n^2R\log K\ge -U^*$.

The next claim verifies {\em(P1)} and thus completes the proof of Theorem~\ref{thm:reduction}.

\begin{lemma}\label{lem:mu-bound}
Either the problem is not feasible, or $\mu_k\le U^*$ holds for any $k\in B\cup G$ in arbitrary $\Delta$-phase.
\end{lemma}
\begin{proof}
Recall from Section~\ref{sec:sink} that we solve the sink version by reducing it to the symmetric algorithm with
$M_t=1$ and $M_i=\lceil 2U^*/\varepsilon\rceil+1$.
Since  $\mu_i$ is non-decreasing, these  values converge to some limits $\mu_i^*\in \mathbb{R}\cup\{\infty\}$.
We have $0\le f_{ij}\le u_{ij}$ on all arcs $ij\in E$ in every phase, and therefore we can choose an infinite subset $T'\subseteq \mathbb{N}$ so that
all $f_{ij}$'s converge if we restrict ourselves to iteration numbers in $T'$. As in Theorem~\ref{thm:delta-opt}, it can be easily verified that the limit
$f^*$ is an optimal solution to the symmetric version with conservative labeling $\mu^*$. 

As in Section~\ref{sec:sink}, if $\kappa_f>2U^*$, then the sink version is not feasible, and otherwise  $f^*$ is also optimal to the sink version. 
In the latter case, $f^*_{it}>0$ for arbitrary $i\in B$ must hold as otherwise $e_t=-\infty$ in the sink version gives infeasibility.
Recall also that the symmetric version was defined such that $e_t<0$ holds for every feasible solution and therefore $\mu^*_t=1$. 

Consider now an arbitrary $i\in B$.
Both $it,ti\in E_f$ (it is easy to verify that $f_{it}=v_{it}$ is impossible),
and therefore $\frac{\mu^*_i}{\mu^*_t}\frac{m_i}{f^*_{it}}=1$ must hold by the conservativity of $\mu^*$. 
This means $\mu^*_i=f^*_{it}/{m_i}\le U\le U^*$. 

Finally, let $j\in G$. Then for an arbitrary arc $ji\in E$, $ji\in E_f$ easily follows, and therefore conservativity gives
$\frac{\mu^*_j}{\mu^*_i}U_{ji}\le 1$, which implies $\mu^*_j\le U^*$.
\end{proof}

Finally, we remark that if we apply this algorithm to linear Fisher markets ($c\equiv 0$), the algorithm runs in a fundamentally different way as \cite{Devanur08} or \cite{Orlin10}. While both these algorithms increase the prices, our algorithm works the other way around: it starts with the highest possible prices, and decreases them.

}

\section{Discussion}\label{sec:disc}
We have given the first polynomial time combinatorial algorithms for both the symmetric and the sink formulation of the concave generalized flow problem. Our algorithm is not strongly polynomial. In fact, no such algorithm is known already for the linear case: it is a fundamental open question to find
a strongly polynomial algorithm for linear generalized flows. If resolved, a natural question could be to devise a strongly polynomial algorithm for some class of convex generalized flow problems, analogously to
the recent result \cite{Vegh11b}, desirably including the market and Nash bargaining applications.

Linear  Fisher market is also captured by \cite{Vegh11b}. A natural question is if
there is any direct connection between our model and the
convex minimum cost flow model studied in \cite{Vegh11b}.
Despite certain similarities, no reduction is known in any
direction. Indeed, no such reduction is
 known even between the linear special cases, that is, generalized flows and minimum-cost circulations.
In fact, the only known market setting captured by both is linear Fisher. 
Perfect price discrimination and ADNB are not known to be reducible to flows with convex objective. In contrast, spending constraint utilities \cite{Vazirani10spending} are not known to be captured by our model, although they are captured by the other.

\iftoggle{full}{
\medskip

As discussed in Section~\ref{sec:overview-concavegen}, it seems difficult to extend  any generalized flow algorithm having separate cycle cancelling and flow transportation subroutines. While this includes the majority of combinatorial algorithms, there are some exceptions. Goldberg, Plotkin and Tardos \cite{Goldberg91}
gave two different algorithms: besides \textsc{Fat-Path}, they also presented \textsc{MCF}, an algorithm that uses a minimum-cost circulation algorithm directly as a subroutine. Hence for the concave setting, it could be possible to develop a similar algorithm using a minimum concave cost circulation algorithm, for example \cite{Hochbaum90} or \cite{Karzanov97} as a black box.

Another approach that avoids scaling is \cite{Wayne02} for minimum-cost generalized flows and \cite{Restrepo09} for generalized flows: these algorithms can be seen as extensions of the cycle cancelling method, extending minimum mean cycles to GAP's in a certain sense. While it does not seem easy, it might be possible to develop such an algorithm for concave generalized flows as well.

In defining an $\varepsilon$-approximate solution for the sink version of concave generalized flows, we allow two types of errors, both for the objective and for feasibility.
A natural question is if either of these could be avoided. While the value oracle model as we use it, seems to need feasibility error, it might be possible to avoid it using a 
stronger oracle model as in \cite{Karzanov97}. One might also require a feasible solution as part of the input, as a starting point to maintain feasibility (For example if all lower bounds and node demands are 0 and $\Gamma_{ij}(0)=0$ on all arcs $ij$, then  $f\equiv 0$ is always feasible).
}

\subsection*{Acknowledgments}
The author is grateful to Vijay Vazirani for several fruitful discussions on market equilibrium problems.
\nottoggle{full}{
\section*{Appendix}
\begin{proof}[Proof of Lemma~\ref{lem:modify-Delta-nonlin}]
Consider a $\Delta/2$-fat arc $ij$ with $\theta_{\Delta/2}^\mu(ij)>1$ for $f$, that is,
\begin{equation}
 \Gamma^{-1}_{ij}\left(\Gamma_{ij}(f_{ij})+\frac\Delta2\mu_j\right)-f_{ij} <\frac{\Delta}2\mu_i.\label{eq:violate}
\end{equation}
There are two possible scenarios: {\em (a)} $ij$ was not $\Delta$-fat, that is,
\begin{equation}
\frac \Delta2\mu_j\le \Gamma_{ij}(u_{ij})-\Gamma_{ij}(f_{ij}) \le \Delta\mu_j,\label{eq:ujel}
\end{equation}
or {\em (b)} $ij$ was also a $\Delta$-fat arc. Then by $\Delta$-conservativity,
\begin{equation}
 \Gamma^{-1}_{ij}(\Gamma_{ij}(f_{ij})+\Delta\mu_j)-f_{ij} \ge \Delta\mu_i.\label{eq:regiel}
\end{equation}
In both cases, let us define
\[
\bar f_{ij}=\Gamma^{-1}_{ij}\left(\Gamma_{ij}(f_{ij})+\frac \Delta2 \mu_j\right).
\]
$\Delta/2$-fatness of $ij$ guarantees that this is well-defined.
%
In case {\em (a)}, we claim that $ij$ is not $\Delta/2$-fat for $\bar f$.
Indeed, 
\[
\Gamma_{ij}(u_{ij})-\Gamma_{ij}(\bar f_{ij})=\Gamma_{ij}(u_{ij})-\left(\Gamma_{ij}(f_{ij})+\frac \Delta2 \mu_j\right)< \frac \Delta2 \mu_j.
\]
The last inequality follows by the second part of (\ref{eq:ujel}).
In case {\em(b)}, we claim that if 
$ij$ is a $\Delta/2$-fat arc for $\bar f$ then $\theta_{\Delta/2}^\mu(ij)\le 1$ must hold for $\bar f$.
Indeed, if we subtract (\ref{eq:violate}) from (\ref{eq:regiel}), we get
\[
 \Gamma^{-1}_{ij}\left(\Gamma_{ij}(f_{ij})+\Delta\mu_j\right)-\Gamma^{-1}_{ij}\left(\Gamma_{ij}(f_{ij})+\frac\Delta2\mu_j\right)> \frac{\Delta}2\mu_i,
\]
and by substituting $\bar f_{ij}$, it follows that
\[
\Gamma^{-1}_{ij}\left(\Gamma_{ij}(\bar f_{ij})+\frac\Delta2\mu_j\right)-\bar f_{ij}> \frac{\Delta}2\mu_i,
\]
that is, $\theta_{\Delta/2}^\mu(ij)<1$ for $\bar f$.

 We next show that  if $ji$ is also a $\Delta/2$-fat arc for $\bar f$, then
$\theta_{\Delta/2}^\mu(ji)\le 1$ holds for $\bar f$. Indeed, 
$\theta_{\Delta/2}^\mu(ji)\le 1$  for $\bar f$ can be written in the  form 
\[
{\Gamma_{ij}(\bar f_{ij})-\Gamma_{ij}\left(\bar f_{ij}-\frac\Delta2\mu_i\right)}\ge \frac\Delta2\mu_j.
\]
Equivalently,
\[
{\Gamma_{ij}(f_{ij})+\frac \Delta2\mu_j-\Gamma_{ij}\left(\bar f_{ij}-\frac\Delta2\mu_i\right)}\ge \frac\Delta2\mu_j.
\]
Subtracting $\frac \Delta2\mu_j$, rearranging and applying the monotone increasing function $\Gamma_{ij}^{-1}$, we get
$f_{ij}\ge \bar f_{ij}-\frac\Delta2\mu_i$, that follows from (\ref{eq:violate}) by substituting $\bar f_{ij}$.

We define $\bar f_{ij}$ the above way whenever $ij$ is a $\Delta/2$-fat arc with $\theta_{\mu}(ij)>1$. (As a simple consequence of concavity, this cannot be the case for both $ij$ and $ji$.)
If this does not hold for neither $ij$ nor $ji$, then let $\bar f_{ij}=f_{ij}$. The next claim compares $f_{ij}$ and $\Gamma(f_{ij})$ to $\bar f_{ij}$ and $\Gamma(f_{ij})$.
\begin{claim}
$|\bar f_{ij}^\mu-f^\mu_{ij}|\le \frac \Delta2$ and $|\Gamma^\mu_{ij}(\bar f_{ij}^\mu)-\Gamma^\mu_{ij}(f^\mu_{ij})|\le \frac \Delta2$.
\end{claim}
\begin{proof}
There is nothing to prove if $\bar f_{ij}=f_{ij}$. Assume $f_{ij}$ was increased as above (the statement is equivalent for $ij$ and $ji$).
The first part is identical to  (\ref{eq:violate}).
By the definition of $\bar f_{ij}$,
\[
\Gamma_{ij}(\bar f_{ij})-\Gamma_{ij}(f_{ij})=\Gamma_{ij}(f_{ij})+\frac\Delta2\mu_j-\Gamma_{ij}(f_{ij})= \frac\Delta2\mu_j,
\]
giving the second part.
\end{proof}
For $\Delta/2$-conservativity, we also need to show that $\bar f$ has no $\Delta/2$-negative nodes with $\mu_i>1/M_i$. 
By the above claim, the total possible change of relabeled flow on arcs incident to $i$ is 
$d_i\Delta/2$. A node is nonnegative for  $\Delta$ if $e_i^\mu\ge d_i\Delta$ and for $\Delta/2$ if  $e_i^\mu\ge d_i\Delta/2$. Consequently, a nonnegative node cannot become $\Delta/2$-negative.

Finally, $Ex_{\Delta/2}^\mu(f)\le Ex_{\Delta}^\mu(f)+\sum_{i\in V}d_i\Delta/2$, and  each arc is responsible for creating at most $\Delta/2$ units of new excess.
This gives $Ex_{\Delta/2}^\mu(\bar f)\le Ex_{\Delta}^\mu(f)+\frac m2\Delta$, as required.
\end{proof}
}
\bibliographystyle{abbrv}
\bibliography{concave}

\end{document}